\documentclass[sigconf, nonacm]{acmart}

\usepackage{multirow}
\usepackage{tikz}
\usepackage{amsmath}
\usepackage{filecontents}
\usepackage{enumitem}
\usepackage{subfigure}
\usepackage{xspace}
\usepackage{cleveref}
\usepackage{listings}
\usepackage{todonotes}
\usepackage{xurl}
\usepackage{booktabs}
\usepackage[normalem]{ulem}
\usepackage{tikz}
\usetikzlibrary{arrows.meta,positioning,calc}

\newcommand*\circled[1]{%
  \tikz[baseline=(char.base)]{%
    \node[shape=circle,draw,fill=black,text=white,inner sep=0.5pt,font=\scriptsize] (char) {#1};%
  }%
}

\newcommand{\system}{\textsc{Remora}\xspace}
\newcommand{\policy}{SFS\xspace}
\newcommand{\router}{\textsf{Coordinator}\xspace}
\newcommand{\worker}{\textsf{Worker}\xspace}
\newcommand{\workers}{\textsf{Workers}\xspace}
\newcommand{\consensus}{\textsf{Consensus}\xspace}

\usepackage[most]{tcolorbox}
\tcbset{ 
    highlight/.style={
            enhanced,                 
            breakable,                
            colback=gray!10,          
            frame hidden,             
            boxrule=0pt,              
            borderline west={2pt}{0pt}{black}, 
            left=2mm, right=2mm, top=1mm, bottom=1mm,
        }
}

\newcommand{\takeaway}[2]{%
    \begin{tcolorbox}[highlight]
        \textbf{Takeaway #1}: \textit{#2}
    \end{tcolorbox}
}

\newcommand{\insight}[2]{%
    \begin{tcolorbox}[highlight]
        \textbf{Insight #1}: {#2}
    \end{tcolorbox}
}

\usepackage[linesnumbered,boxruled,vlined]{algorithm2e}

\RestyleAlgo{ruled}         
\SetAlgoNlRelativeSize{-1}  

\newcommand{\algline}[2]{Algorithm~\ref{#1}, Line~\ref{#2}}

\usepackage[font=small,labelfont=bf]{caption}
\newtheorem{theorem}{Theorem}
\newtheorem{lemma}[theorem]{Lemma}
\newtheorem{remark}[theorem]{Remark}

\DontPrintSemicolon
\SetKwProg{Proc}{Procedure}{}{}
\SetFuncSty{normalfont}

\SetKwFunction{OnProposedBlock}{\textsc{OnProposedBlock}}
\SetKwFunction{OnCommittedBlock}{\textsc{OnCommittedBlock}}
\SetKwFunction{PlanBlock}{\textsc{PlanBlock}}
\SetKwFunction{RunStateless}{\textsc{RunStateless}}
\SetKwFunction{Dispatch}{\textsc{Dispatch}}
\SetKwFunction{OnStateless}{\textsc{OnStateless}}
\SetKwFunction{OnStateful}{\textsc{OnStateful}}
\SetKwFunction{EnsureLocal}{\textsc{EnsureLocal}}
\SetKwFunction{OnEpochTrigger}{\textsc{OnEpochTrigger}}
\SetKwFunction{OnWorkerFailure}{\textsc{OnWorkerFailure}}
\SetKwFunction{ScaleOut}{\textsc{ScaleOut}}
\SetKwFunction{ScaleIn}{\textsc{ScaleIn}}
\SetKwFunction{EnsureOwnership}{\textsc{EnsureOwnership}}

\SetKwFunction{MatchesProposal}{\texttt{MatchesProposal}}
\SetKwFunction{Adopt}{\texttt{Adopt}}
\SetKwFunction{AssignVersions}{\texttt{AssignVersions}}
\SetKwFunction{BuildDependencies}{\texttt{BuildDependencies}}
\SetKwFunction{Partition}{\texttt{Partition}}
\SetKwFunction{PlaceUnits}{\texttt{PlaceUnits}}
\SetKwFunction{ScheduleStateless}{\texttt{ScheduleStateless}}
\SetKwFunction{PlanVers}{\texttt{PlanVers}}
\SetKwFunction{PlanPlace}{\texttt{PlanPlace}}
\SetKwFunction{PlanSources}{\texttt{PlanSources}}
\SetKwFunction{SendToWorker}{\texttt{SendToWorker}}
\SetKwFunction{Validate}{\texttt{Validate}}
\SetKwFunction{PublishToken}{\texttt{PublishToken}}
\SetKwFunction{AwaitAllAvailable}{\texttt{AwaitAllAvailable}}
\SetKwFunction{AwaitToken}{\texttt{AwaitToken}}
\SetKwFunction{Invalid}{\texttt{Invalid}}
\SetKwFunction{ConsumeNoOp}{\texttt{ConsumeNoOp}}
\SetKwFunction{Read}{\texttt{Read}}
\SetKwFunction{Execute}{\texttt{Execute}}
\SetKwFunction{Write}{\texttt{Write}}
\SetKwFunction{Ack}{\texttt{Ack}}
\SetKwFunction{Owner}{\texttt{Owner}}
\SetKwFunction{HasVersion}{\texttt{HasVersion}}
\SetKwFunction{AuthorizeTransfer}{\texttt{AuthorizeTransfer}}
\SetKwFunction{LatestLeq}{\texttt{LatestLeq}}
\SetKwFunction{Fetch}{\texttt{Fetch}}
\SetKwFunction{PeerFetch}{\texttt{PeerFetch}}
\SetKwFunction{FetchFromRouter}{\texttt{FetchFromRouter}}
\SetKwFunction{SetOwner}{\texttt{SetOwner}}
\SetKwFunction{RequestSnapshot}{\texttt{RequestSnapshot}}
\SetKwFunction{CollectAndPersist}{\texttt{CollectAndPersist}}
\SetKwFunction{LatestSnapshotEpoch}{\texttt{LatestSnapshotEpoch}}
\SetKwFunction{SpawnWorkers}{\texttt{SpawnWorkers}}
\SetKwFunction{Seed}{\texttt{Seed}}
\SetKwFunction{LogRange}{\texttt{LogRange}}
\SetKwFunction{Replay}{\texttt{Replay}}
\SetKwFunction{PromoteReplacement}{\texttt{PromoteReplacement}}
\SetKwFunction{SpawnWorker}{\texttt{SpawnWorker}}
\SetKwFunction{BiasStateless}{\texttt{BiasStateless}}
\SetKwFunction{GraduallyPlaceStateful}{\texttt{GraduallyPlaceStateful}}
\SetKwFunction{StopPlacingStateful}{\texttt{StopPlacingStateful}}
\SetKwFunction{WaitForLeasesToTransferOrExpire}{\texttt{WaitForLeasesToTransferOrExpire}}
\SetKwFunction{Shutdown}{\texttt{Shutdown}}

\begin{document}

\date{}

\title{\system: Scale-out Deterministic Execution for Smart Contracts}

\author{
  Zhengqing Liu$^{1}$, Alberto Sonnino$^{2,3}$, Igor Zablotchi$^{2}$, Eleftherios Kokoris Kogias$^{2}$, Marios Kogias$^{1}$
}

\affiliation{
  \institution{$^{1}$Imperial College London, $^{2}$Mysten Labs, $^{3}$University College London}
  \country{}
}


\begin{abstract}
Modern blockchains rely on a modular architecture that decouples consensus from execution.
Recent advances in consensus algorithms have shifted the bottleneck to the execution layer, which must deterministically follow the consensus order and handle increasingly complex, compute-intensive smart contracts.
We identify that single-node validators cannot keep up, motivating the need for a scale-out design.

We design \system, a scale-out smart contract execution engine.
\system adopts an efficient asymmetric architecture with centralized transaction dispatching and distributed execution, and depends on an object versioning scheme with a strict ownership model to guarantee deterministic scale-out execution.
\system achieves up to 3$\times$ throughput improvement compared to state-of-the-art deterministic execution schemes, scales up to 250k TPS, matching modern consensus performance, and reduces latency by up to 5$ms$.
We also show that \system elastically adapts to bursty workloads and dynamic access patterns using real-world traces.
\system's main performance benefits come from a novel stateless-stateful separation during smart contract execution, which overlaps the execution of state-independent tasks with consensus, and a new locality-aware and load-balanced scheduling scheme.
\end{abstract}

\maketitle
\pagestyle{plain}

\section{Introduction}\label{sec:intro}

Modern blockchains are operated by \textit{validators}, which adopt a layered architecture (Figure~\ref{fig:blockchain-arch}): the \textit{consensus} layer establishes a globally ordered sequence of transactions, and the \textit{execution} layer performs transaction execution.
While enormous progress has been made in scaling consensus, with production systems today sustaining 200k–300k transactions per second (TPS)~\cite{mysticeti,shoal++}, the execution layer has failed to keep pace.
The need for faster execution is amplified by the rise of smart contracts, which drive new demands from decentralized applications spanning finance, gaming, and identity~\cite{dappradar2023games, okta2024decentralizedidentity, qin2023blockchain, heavy-stateless-motivation}.
Thus, execution, not consensus, now defines the scalability frontier for blockchains.

Despite advances in exploiting multi-core parallelism on smart contract VMs~\cite{pevm, parallelevm, crystality, blockstm}, \textit{scale-up} execution alone cannot meet the computational demands of modern smart contracts.
As smart contract logic and cryptographic authentication grow more expensive, due to techniques like zero-knowledge proofs and post-quantum cryptography~\cite{algorandPostQuantum2025,fastcrypto, a16z_quantum_blockchains}, scaling-up becomes impractical as computational needs exceed the capacity of single-node validators.
Addressing this bottleneck requires a shift to \textit{\textbf{scale-out}} designs that distribute validator's execution across multiple machines.


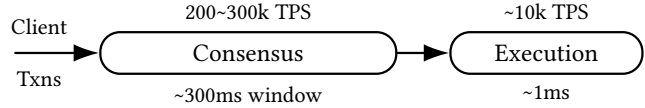
\begin{figure}[t]
  \centering
  \resizebox{\columnwidth}{!}{%
  \begin{tikzpicture}[
      font=\footnotesize,
      line cap=round,
      line join=round,
      >={Latex[length=3.2mm,width=2.2mm]},
      block/.style={
        draw,
        thick,
        rounded corners=7pt,
        minimum height=0.52cm,
        align=center,
        inner sep=3pt
      }
    ]

    \node[block, minimum width=3.7cm] (cons) at (4.8,0) {\normalsize Consensus};
    \node[block, minimum width=2.4cm, right=0.65cm of cons] (exec) {\normalsize Execution};

    \draw[->,thick] ($(cons.west)+(-1.05,0)$) -- (cons.west);
    \draw[->,thick] (cons.east) -- (exec.west);

    \node[anchor=south,font=\small] at ($(cons.west)+(-0.75,0.12)$) {Client};
    \node[anchor=north,font=\small] at ($(cons.west)+(-0.75,-0.12)$) {Txns};

    \node[font=\small] at ($(cons.north)+(0,0.24)$) {200\textasciitilde300k TPS};
    \node[font=\small] at ($(exec.north)+(0,0.24)$) {\textasciitilde10k TPS};

    \coordinate (wL) at ($(cons.south west)+(0,-0.10)$);
    \coordinate (wR) at ($(cons.south east)+(0,-0.10)$);

    \node[font=\small] at ($0.5*(wL)+0.5*(wR)+(0,-0.14)$) {\textasciitilde300ms window};

    \node[font=\small] at ($(exec.south)+(0,-0.22)$) {\textasciitilde1ms};

  \end{tikzpicture}%
  }
  \caption{\textbf{Validator architecture in modern blockchains.}}
  \label{fig:blockchain-arch}
\end{figure}

A core requirement of blockchain execution that makes the problem more challenging compared to prior work on distributed transaction processing~\cite{redt,h-store, turbodb}, is \textit{\textbf{strict determinism}}~\cite{spectrum}, i.e., preserving the total order established by the consensus layer during execution even in the presence of parallel or distributed executors.
In blockchains, replicas are mutually untrusted, and transaction order often carries financial significance in applications like auctions and flash loans.
Enforcing the consensus-established order is essential because it guarantees fairness, preserves transparency, and ensures that execution remains verifiable and tamper-proof~\cite{chaindash, auction}.


In this work, we ask: \textit{how to efficiently scale-out smart contract execution while preserving strict determinism?}
Revisiting prior work on deterministic transactional systems, we identify two fundamental design pitfalls.
First, all prior deterministic databases adopt a \textit{symmetric} architecture, that duplicates the full control stack, i.e., sequencing and scheduling all transactions, on every node, while only performing execution for a fraction of them~\cite{calvin, delaypart, qstore, slog, detock}.
This redundancy inflates the coordination overhead and leads to limited scalability.
Second, there is currently no distributed transaction-scheduling scheme that jointly satisfies strict determinism, locality awareness, and load balance, properties that are necessary for efficient scale-out smart contract execution.
Load-driven schemes trigger excessive state migration and coordination for distributed transactions, while locality-driven ones concentrate records on a single node~\cite{leap, zeus}.
Recent work attempts to balance both~\cite{delaypart, hermes}, but often breaks strict determinism via reordering, which is unacceptable to blockchain systems.

Beyond identifying design pitfalls in prior work, we further observe two new domain-specific opportunities by leveraging block\-chain characteristics.
First, smart contract execution naturally comprises two steps: a \emph{stateless} one, performing verification and authentication, which is compute-intensive but independent of shared state; and a \emph{stateful} one, performing business logic of smart contracts, which thus requires access to the blockchain state but tends to be lightweight.
Accessing shared state introduces dependencies and ordering constraints across transactions.
However, we argue that the stateless parts have no such dependencies and can be executed independently on any available compute resources, outside the critical dependency boundaries.
Second, the stability of existing consensus algorithms introduces a \emph{window} of opportunity between transaction proposal and commitment with high predictability in proposal outcomes, which allows for speculative off-path execution to perform useful work ahead of time~\cite{forerunner}.
We pinpoint the parts of the validator execution that can be efficiently performed during the consensus window, without incurring excessive state transfers and cascading aborts in a scale-out setting.

Driven by these findings, we design \system, a scale-out execution engine for blockchain systems that preserves strict determinism.
\system adopts an \emph{asymmetric} architecture, with a single \router managing a pool of \workers to decouple scheduling from execution.
The \router receives the globally ordered transaction stream from consensus and dispatches each transaction to a \worker for execution.
To enforce strict determinism, \system introduces \textit{per-object versioning} and a \textit{lease-based ownership} model: each versioned object is exclusively owned by either the \router or a single \worker at any point in time.
Each transaction is annotated with an explicit read/write set, and the \router assigns object versions in consensus order.
When a transaction requires objects owned by a different \worker, \system transfers leases and object state accordingly.
This design simplifies the enforcement of determinism and, by allowing dynamic flow of object ownership, enables scheduling optimizations that collocate related transactions and improve locality.

\system decouples the execution of stateless and stateful parts, proposing \textit{pre-consensus stateless execution and stateful scheduling} to gain both latency and throughput benefits.
The separation allows compute-heavy stateless parts to harness available resources within the cluster.
To handle stateful parts under contended and skewed workloads, \system introduces a \textit{subgraph-first scheduling (\policy)} scheme that focuses on scheduling disconnected subgraphs within the transaction dependency graph, jointly considering transaction ordering, load distribution, and data locality.
To tolerate worker failures, \system further employs periodic snapshots to persist batched updates.
\system also implements a \router-driven elasticity scheme that adds or removes nodes to the cluster and depends on lazy state transfers to the new nodes.

We evaluate \system using transactional benchmarks and real-world smart-contract traces, showing up to \textbf{$3\times$} throughput improvement compared to state-of-the-art deterministic execution scheme~\cite{hermes}, while effectively masking the latency of both stateless execution and stateful scheduling in consensus window, thus reducing end-to-end latency by up to 5ms.
\system's asymmetric architecture also demonstrates higher efficiency than traditional deterministic databases~\cite{calvin, detock} as system scales, thus reducing deployment cost.
Our implementation sustains over 250k TPS even in the presence of split execution and the computationally heavy \policy scheduling, matching the modern consensus throughput.

This paper makes the following contributions:
\begin{itemize}[noitemsep,topsep=1pt,leftmargin=*]
  \item
    Two blockchain-specific design opportunities (\S\ref{sec:insight}), i.e., decoupling stateless and stateful execution, and enabling pre-consensus stateless execution and stateful scheduling to improve performance in a scale-out setting (\S\ref{sec:design:separation}, \S\ref{sec:design:consensus-window}).
  \item
    An object-versioning and lease-based ownership model to preserve strict determinism in a distributed setting and simplify deterministic parallel execution in each \worker (\S\ref{sec:design:determinism}).
  \item
    \policy, a subgraph-first scheduling scheme that jointly considers load, locality, and dependency ordering for efficient execution of smart contract workloads (\S\ref{sec:design:policy}).
  \item
    \system, a scale-out execution engine for smart contracts that adopts an asymmetric architecture, improving efficiency over prior deterministic distributed systems and addressing blockchain-specific requirements (\S\ref{sec:design:arch}).

\end{itemize}

\section{Background and Motivation}\label{sec:background}
In this section, we describe the high-level architecture of modern blockchains, identify the scalability bottleneck, and motivate the need for scale-out.
We then outline the unique requirements of blockchain infrastructure and explain why existing (non-) deterministic distributed transaction systems fall short of meeting them.

\subsection{Blockchain Architecture and Bottlenecks}
\begin{figure}[t]
    \centering
    \includegraphics[width=\columnwidth]{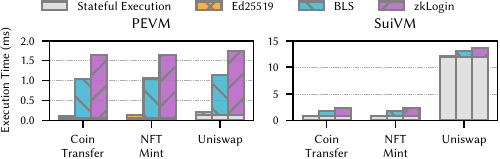}
    \caption{
        Execution breakdown of common smart contracts.
    }
    \label{fig:example-contracts}
\end{figure}

Modern blockchains typically adopt modular architectures that decouple consensus from execution~\cite{solana-code,sui-code,aptos-code}.
Validators first organize transactions in blocks and propose them to the network.
These blocks are then ordered by a consensus protocol before being executed deterministically by the validators' execution engine.
This decoupling enables independent evolution of the two layers and allows diverse combinations of technologies.
For example, both Aptos~\cite{aptos} and Cronos~\cite{cronos} adopt BlockSTM~\cite{blockstm} for execution, despite relying on different consensus protocols.

This architectural decoupling also means that either component, consensus or execution, can independently become a bottleneck.
Although early consensus protocols only scaled to a few hundreds or thousands of TPS~\cite{pbft}, modern designs now sustain over 300k TPS by allowing validators to propose transactions in parallel and leveraging state-of-the-art BFT protocols~\cite{mysticeti, shoal++}.
As a result, \textit{the bottleneck has now shifted to execution}.
Despite extensive parallelization efforts~\cite{sui,heimbach2023defi,hay2024batch, doradd, blockstm},
end-to-end throughput in practice remains at tens of kTPS due to the smart contract complexity and the overhead of the underlying VMs~\cite{sui-lutris}.

To better understand the execution bottleneck, Figure~\ref{fig:example-contracts} shows the average execution time of several smart contracts on two virtual machines~\cite{pevm, sui-code}.
Depending on the contract logic and the verification method, the execution can take up to several milliseconds of compute time.
Even under ideal assumptions with a modest execution time of 1 ms per transaction, uniform load with limited inter-transaction dependencies, Little’s Law implies that sustaining 300k TPS (the throughput of state-of-the-art consensus protocols) would require roughly 300 cores dedicated to execution.
This far exceeds the capacity of typical validator hardware~\cite{sui-min-specs,solana-min-specs,aptos-min-specs}.


A broad line of work tackles such bottlenecks via \emph{blockchain sharding}, where the blockchain state (and hence execution) is partitioned across shards maintained by different validator committees in parallel~\cite{sharper,sigmod19-sharding,omniledger,byshard,arete,gridb,marlin}.
However, sharding introduces recurring challenges.
First, it relies on subcommittee formation through random sampling which needs stronger adversarial assumptions to keep every shard secure with high probability~\cite{omniledger,sigmod19-sharding,arete}.
Second, transactions that span shards incur substantial overhead due to Byzantine-resilient cross-shard coordination (often akin to atomic commit), which can dominate performance under realistic cross-shard access patterns~\cite{sharding-security}, since state migration requires subcommittee-to-subcommittee BFT coordination.
These challenges make sharding operationally complex, and recent production roadmaps have decided to shift away from sharding~\cite{ethereum_roadmap}.

We argue there is an alternative scale-out approach: \emph{scale out execution within a single validator}.
Rather than repartitioning state across multiple validator committees, we can focus on the execution bottleneck of a single validator and distribute the workload across a cluster of machines, while preserving the consensus protocol, validator set, and trust assumptions.
This design allows a validator, which is a single administrative and trust domain, to scale its execution capacity elastically with available compute resources, narrowing the widening gap between consensus throughput and single-node execution.

\takeaway{1}{The computational need of modern smart contracts goes beyond the capacity of single-node validators, making the scale-out design necessary.
Scale-out intra-validator execution can retain the trust model and operational requirements of modern blockchains, unlike inter-validator sharding approaches.
}

\subsection{Design Requirements and Prior Work Limitations}
\label{sec:design-requirements}
Focusing on a scale-out architecture of a \emph{single validator}, we enumerate the key requirements that guide such a design.
To do so, we revisit related work across blockchains and distributed databases/systems to highlight where prior approaches fall short and to surface the practical considerations for real deployments.

\noindent\underline{\textbf{R1: Strict Determinism.}}
Preserving the consensus-established order, i.e., strict determinism~\cite{spectrum}, is the fundamental correctness requirement for blockchain execution.
Given a globally agreed-upon transaction order, a strictly deterministic execution layer guarantees that the resulting effects are identical to those of sequential execution.
This corresponds to \emph{strictly-deterministic serializability}~\cite{spectrum}: for an order of transactions $O=\langle T_1,\dots,T_n\rangle$, the execution effect is equivalent to sequentially executing and committing $O$ in that order.
Other prior work also formalizes the specific total order into the correctness target known as \textit{Byzantine Ordered Consensus}~\cite{byzantine-ordered-consensus, ordered-consensus-equal}.
This requirement sharply differentiates blockchains from conventional deterministic systems within a trusted administrative domain, where controlled reordering can be used to improve performance by reducing dependencies or object transfers~\cite{aria, hermes, delaypart}.


In a blockchain setting, strict adherence to the consensus order is critical for two reasons.
First, replicas are operated by mutually untrusted parties, and the established order of transactions can have direct financial consequences in order-sensitive applications (e.g., auctions, flash minting)~\cite{chaindash, auction}.
Second, blockchains are designed as permanent, decades-spanning ledgers requiring continuous public auditability.
If execution engines were allowed to apply different reorderings over time, validators would need to record additional per-block metadata to reproduce the chosen order, complicating verification and weakening the replayability.
Strict adherence to consensus order is thus the current industry standard.

\noindent\underline{\textbf{R2: Efficiency and Cost-Effectiveness.}}
\label{sec:background:efficiency}
A scale-out execution layer should convert added machines into execution throughput, not duplicated control-plane work.
Calvin~\cite{calvin} pioneered distributed deterministic execution by enforcing a global transaction order to eliminate two-phase commit.
In Calvin, every node receives the fully replicated, ordered batch of transactions, even though each node manages only a partition of the object space.
To maintain determinism, nodes process only the transactions relevant to their partition, strictly following the global order.
Transactions execute in-place on partitioned state, with remote reads and coordination handled during batched phases.

Calvin-style designs have been widely followed by later systems~\cite{hermes, delaypart, qstore, slog, detock, peep}, which keep a \emph{symmetric} architecture: every node sequences, schedules, and executes.
This preserves determinism but duplicates work, i.e., each node processes the full batch of transactions, runs the same dependency analysis, and schedules, even though it executes only a subset.
The result is wasted CPU, higher coordination overhead, and higher per-node provisioning cost (full metadata/control stack on every machine).
As the cluster or throughput grows, these replicated costs dominate, so scale-out shows diminishing returns (quantified in \S\ref{sec:eval-arch}).

\noindent\underline{\textbf{R3: Workload Adaptiveness.}}
\label{sec:background:workload-adaptiveness}
A scale-out execution layer should adapt online to workload skew and hotspots, minimizing distributed transactions by favoring locality, while still balancing loads.
Prior work shows that performance in distributed databases is highly sensitive to data partitioning~\cite{hermes, delaypart}, as it determines the fraction of distributed versus single-node transactions.
Although offline profiling or periodic repartitioning can improve locality~\cite{clay, delaypart}, real workloads are dynamic and hard to predict~\cite{morphdag,schain,flexchain}.
Even workload-driven repartitioning and live object migration~\cite{morphosys, schism} can fall short in complex deployments, as prior work~\cite{hermes} highlights.
Hence, recent systems combine repartitioning with on-demand migration so that objects flow across nodes according to the workload needs as part of transaction processing~\cite{zeus, leap, hermes}.

However, freely allowing object migrations across nodes can lead to further inefficiencies.
On one hand, load-driven policies, which try to equalize the load among nodes, trigger frequent state transfers and “ping-pong” effects~\cite{chimera, hermes}, causing stalls and inefficiency.
On the other hand, locality-driven policies, which seek to minimize object transfers, collapse active records onto a single node in the presence of skewed workloads, as shown in \S\ref{sec:eval-policy}.
Skewed workloads are the norm in modern blockchains~\cite{ethereum-conflicts-graphed,parallelevm}.
Some designs consider both locality and load, but are either non-deterministic~\cite{rtsfaas} or heavily rely on reordering~\cite{hermes,delaypart}, thus preventing their use in blockchain systems that require strict determinism.

\noindent\underline{\textbf{R4: Fault Tolerance.}}
Scale-out execution introduces additional failure modes: execution nodes may crash or be temporarily unavailable.
Our goal is to scale \emph{intra-validator} execution, without changing the consensus layer.
BFT remains provided by the blockchain’s consensus across validators, and our design should not change its fault threshold or safety guarantees.
Given that one validator is a single administrative and trust domain, we target \emph{crash fault tolerance} for execution nodes.
Specifically, we aim to preserve availability and avoid reducing a validator's mean-time-to-failure (MTTF) due to the increased number of components that can fail independently, rather than to defend against intra-validator compromise.

\noindent\underline{\textbf{R5: Elastic Autoscaling.}}
Production workloads are bursty and time-varying.
Thus, a practical execution layer should be elastic to scale resources up/down and adapt placement online without disruptive global reconfiguration.
Past permissioned systems~\cite{flexchain, adachain, fab} emphasize the need for elasticity, yet they focus on the \emph{whole chain} and make substantial and potentially non-compatible changes to the execution, architecture, and all validators.
This work instead targets \emph{single-validator} elasticity arguing that each validator should scale its execution capacity independently without changing the consensus layer or affecting other validators.

\takeaway{2}{Scale-out blockchain execution requires a careful co-design for strict determinism, architecture efficiency, workload adaptiveness, fault tolerance, and elasticity as independent design choices might lead to requirement violations.}

\section{\system Insights}
\label{sec:insight}

Before diving into the \system design, we highlight two crucial observations that are unique to blockchain systems and substantially influence the proposed architecture.

\subsection{Stateless-Stateful Separation}

We observe that smart contract execution naturally decomposes into two components: \emph{stateless} operations, dominated by cryptographic verification of transaction authenticity, and \emph{stateful} operations, which execute the contract logic and modify the blockchain state.
This separation creates an opportunity to scale stateless work independently of stateful execution.
Stateless work is completely independent of the smart contract logic or the blockchain state, imposing no ordering or placement constraints.
Hence, it can be executed on any available compute resource in the validator cluster.

\Cref{fig:example-contracts} breaks down the execution time of three representative smart contracts: a simple transfer, an NFT mint, and a Uniswap~\cite{lo2021uniswap} trade.
Two smart contract VMs are used: SuiVM~\cite{sui-code}, a production-grade MoveVM powering the Sui~\cite{sui} blockchain, and PEVM~\cite{pevm}, a recent parallel EVM implementation written in Rust.
We pair each stateful part with one of three primary authentication methods~\cite{fastcrypto} in modern blockchains: (1) simple signatures (e.g., Ed25519~\cite{brendel2021provable} or BLS~\cite{bls}), (2) key-less authentication via external identity providers (e.g., zkLogin~\cite{zklogin}), and (3) multi-signatures combining several instances of these methods.
Differences in stateful execution time across VMs reflect their varying runtime overheads and optimizations—lighter VMs reduce execution costs by design, making the stateful component more efficient.

We observe that stateless compute often dominates total execution time of the smart contract: 1\,ms for BLS~\cite{li2023performance,fastcrypto-benchmarks}, 1.6\,ms for zkLogin~\cite{fastcrypto-benchmarks,groth16-benchmark}, and up to 16\,ms for multi-signatures involving ten zkLogin proofs.
While simple signatures remain most common, key-less authentication is rapidly gaining adoption.
We observe roughly 10k zkLogin-based transactions per day on a production blockchain validator~\cite{sui}.
Multi-signatures are also used, albeit less frequently, with about three transactions per minute requiring verification of multi-signatures comprising two to ten signatures or zkLogin proofs.
As on-chain workloads expand to include post-quantum cryptography, the stateless portion of execution is expected to become increasingly compute-intensive~\cite{heavy-stateless-motivation,algorandPostQuantum2025, a16z_quantum_blockchains}.
\insight{1}{Smart contract execution consists of a stateless and a stateful part, with the stateless portion being increasingly compute-intensive yet runnable without ordering or placement constraints on any available compute resource.
}

\subsection{Consensus Window}

Prior work~\cite{forerunner} has identified there is a window between when a transaction is proposed and when it is finally committed that can be leveraged for speculative execution to improve performance of the blockchain.
In state-of-the-art production-grade consensus systems~\cite{mysticeti}, this window can last over 300ms, which is a substantial amount of time compared with smart contract execution time.

To evaluate the opportunity, we study existing blockchain consensus algorithms and deployed systems and observe the following.
Existing consensus protocols exhibit a high degree of predictability in practice: proposals are committed in the expected order in the overwhelming majority of cases.
Deviations occur only when
(i) a block proposal is invalid (e.g., due to faulty validators), or
(ii) network asynchrony delays or reorders commitments.
Empirical evidence from Sui's L1 blockchain infrastructure shows that 98.66\% of blocks are committed directly in the proposed order~\cite{sui}.
Moreover, after a block proposal collects a quorum of votes~\cite{pbft} in the first phase of the protocol~\cite{pbft,hotstuff,bullshark,mysticeti}, only extreme network asynchrony can prevent it from being committed as expected.

Leveraging the consensus window becomes considerably more challenging in the context of scale-out execution.
Prior work~\cite{forerunner,seer} exploits this window for speculative smart contract execution.
This is feasible in their setting because they assume a single-node validator, which holds the entire blockchain state locally.
In a scale-out setting, however, state is partitioned across nodes.
Speculative execution would then require speculative object transfers across partitions to resolve dependencies.
These transfers not only impose significant communication cost, but also create the risk of cascading aborts: a single misprediction in the final commit order can invalidate multiple speculative transfers and their dependent computation across nodes.
The combined cost of speculative state movement and large-scale rollbacks can easily outweigh the performance gains, rendering naive scale-out execution impractical.

To still harvest the benefits of the consensus window, we identify two parts of validator execution that can be performed speculatively without requiring access to the underlying blockchain state.
These are: (i) the stateless part of a smart contract and (ii) the scheduling logic that determines which transactions each node will execute.
The first requires only the transaction input, while the second relies on placement metadata and prior scheduling decisions.
Crucially, neither depends on the distributed blockchain state.
By only optimistically performing these state-independent components, we capture the benefits of the consensus window without costly transfers or cascading aborts.

\insight{2}{
In a scale-out design, validator tasks that can be efficiently performed in the consensus window are stateless execution and transaction scheduling.}

\section{The \system Design}\label{sec:design}
\begin{figure*}
	\centering
	\includegraphics[width=2\columnwidth]{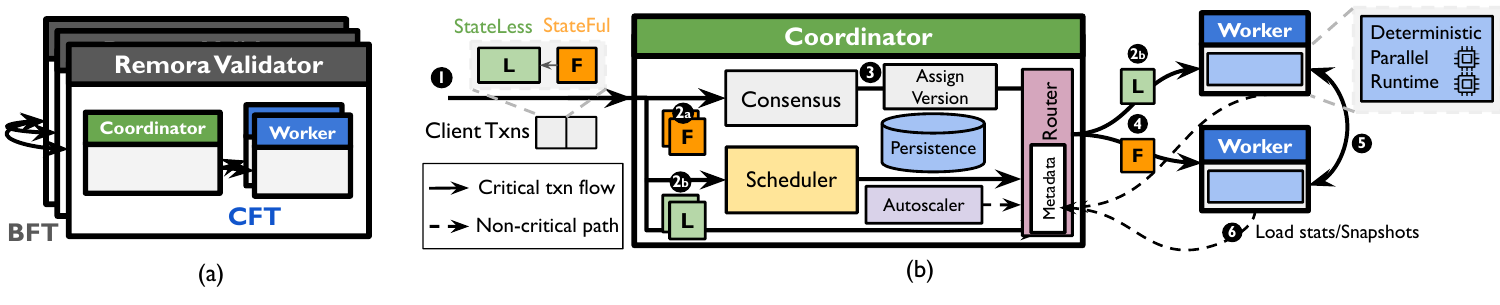}
	\caption{\system Architecture. {\system is a single validator inheriting Byzantine Fault Tolerance (BFT) from the consensus layer (a), and internally scales out execution (b) with a centralized \router and a pool of \workers under Crash Fault Tolerance (CFT).}}
	\label{fig:arch}
	\vspace{-0.4cm}
\end{figure*}

We design \system, a scale-out execution engine for a single validator, to cater to the needs of modern blockchain infrastructure.
\system should:
(i) follow a modular architecture that decouples consensus from smart contract execution and preserve strict determinism;
(ii) leverage stateful and stateless execution separation;
(iii) allow to implement various pre- and post-consensus scheduling algorithms;
(iv) enable seamless elasticity by dynamically adapting to load without upfront costs;
(v) tolerate execution node failures within a trusted validator.
We provide formalized proofs in the \hyperref[appendix]{Appendix}.

\subsection{High-Level Overview}
\label{sec:design:arch}

\noindent\textbf{Architecture.}
\system adopts an \textit{asymmetric} architecture that splits validator functionality across specialized nodes, in contrast to prior symmetric designs (\S\ref{sec:background:efficiency}) where every node redundantly performs ordering, scheduling, and execution.
A single \router node handles sequencing and scheduling, while a pool of \worker nodes split the object state among them and execute smart contracts.
The \router is the only participant in the consensus protocol, from which it receives the ordered sequence of transactions.
Then, based on the scheduling logic, it dispatches transactions to the \workers for deterministic scale-out execution.
The \router also manages elasticity and resource allocation, durably persisting execution results, and coordinating failure recovery.

\noindent\textbf{Workflows.}
Figure~\ref{fig:arch} gives an overview of the main components and the transaction execution workflow in \system.
The \router receives transactions from clients (\circled{1}), and the consensus module runs a sequencing protocol to determine their global order according to the specific blockchain deployment.
\system separates the stateless and stateful components early: it forwards the stateful part to the scheduler for speculative planning (\circled{2a}), and dispatches the stateless part directly to \workers according to the scheduling policy (\circled{2b}).
Both steps leverage the consensus window to overlap with consensus and reduce latency.
After consensus finalizes the transaction sequence, \system assigns object versions to each transaction (\circled{3}) to enforce a deterministic execution schedule preserving the established order.
The \router then routes the stateful components to \workers using the scheduler’s precomputed plan.
\workers execute only the transactions that the \router assigns to them (\circled{4}) and communicate with their peers when remote state is needed (\circled{5}).
They also periodically report load and execution state snapshots (\circled{6}), enabling autoscaling to adjust cluster size based on utilization.
The persistence module in the \router safely stores these snapshots to ensure crash recovery for \workers, alongside the transaction logs determined by the consensus.

This asymmetric design effectively fulfills our design goals.
First, the consensus layer is agnostic to \system, making the architecture reusable across different blockchain deployments.
Second, concentrating scheduling in the \router avoids duplicating scheduling logic at every \worker and enables new policies without \worker changes.
Third, with its global view of cluster utilization, the \router can make accurate autoscaling decisions.
Finally, persisting results asynchronously combined with deterministic execution allows for efficient \worker fault tolerance.


\noindent\textbf{Threat Model.}
\system scales out execution within a single trust domain, the validator.
A \system deployment should therefore be viewed as \textit{a single consensus participant} implemented by a \router and a pool of internal \worker machines.
The \router is the validator's sole externally visible interface to consensus, and any fault suffered by the \router is applied as a total domain-wide failure as far as the blockchain protocol is concerned.
\workers, in contrast, are internal execution resources under the same administrative trust domain as the \router (analogous to threads/cores in a scale-up validator), so we assume they are trusted but crash-stop, i.e., they may fail by crashing, but do not behave maliciously.
\worker failures do not affect the blockchain protocol's security as they only affect the validator's performance. 

\subsection{Enforcing Strict Determinism}
\label{sec:design:determinism}

\system assumes a widely used smart contract model~\cite{rw-assumption, doradd} in which contracts execute atop a flat key-value store and provides ACID transactional guarantees.
Each contract declares its read and write sets upfront, an assumption shown to hold for most transactional workloads (further discussed in~\S\ref{sec:discussion}).
The \router acts as the global owner of all objects but delegates execution by leasing objects to \workers.
Each \worker holds the leased state in memory and periodically reports updates back to the \router.
The \router maintains metadata tracking the current leaseholder for each object.
Based on the scheduling policy, the \router forwards a transaction to a single \worker for execution.
If all objects reside in the delegated \worker, execution runs to completion.
Otherwise, the \worker initiates lease transfers and fetches remote objects from their current owners before execution.

\noindent\textbf{Object Versioning}.
\system's scale-out design introduces challenges in guaranteeing deterministic execution in a distributed setting, beyond the scope of single-node deterministic parallel execution~\cite{doradd, blockstm, pwv}.
To enforce strict determinism in a scale-out manner, \system leverages the \router as a serialization point.
For each transaction in the consensus-ordered batch, the \router assigns a unique version to every accessed object (read or write) by incrementing its per-object counter.
This ensures that all accesses to the same object are consistent with the consensus-established order, guaranteeing strict determinism.
Version assignment requires only a single linear scan over per-transaction objects within the consensus-ordered batch, incurring negligible overhead.
Note that this version is validator-local and does not have to be the same across different validators.
This version-based design is inspired by single-node deterministic parallel systems, such as Bohm~\cite{bohm}, where a single thread assigns object versions before execution.
\system generalizes this principle to distributed execution.

Centralized version assignment simplifies the design of each \textsf{Worker}'s local runtime, which must execute transactions in parallel while preserving determinism.
When the \router dispatches a transaction to a \worker, it explicitly defines the specific object versions the transaction needs to access along with their owner node.
A transaction does not compute dependencies by identifying the last writer of an object; instead, its dependencies are fully determined by the assigned versions.
A transaction becomes runnable once all required versions are available, either locally produced or fetched from remote workers.
It observes the shared state defined by input versions, and produces the next versions for its accessed objects.
The per-object version stream assigned by the \router ensures that execution respects the consensus order without requiring additional coordination or locking.
Furthermore, this naturally enables \textit{inter-block} parallelism~\cite{schain}, as \worker execution proceeds independently of block boundaries.

Object versions have explicit lifetimes, making garbage collection trivial.
The \router assigns a unique version per transaction access, ensuring each version is consumed \textit{exactly once}.
Once the consuming transaction executes, the \worker safely discards the old version.
A new version is created even for reads, which simplifies memory management but limits traditional read-sharing optimizations.
The performance cost of this trade-off is minimal: hot objects in blockchain workloads, such as balances, order books, and NFTs, are update-heavy, and prior measurements show that write-write conflicts are far more prevalent than read-write conflicts~\cite{ethereum-conflicts-graphed}.


\noindent\textbf{Object Ownership and Leasing}.
\system adopts a strict object ownership model: at any point in time, each object resides on a single \worker, which holds an exclusive lease to access and modify it, or the \router.
Static partitioning or fixed sharding performs poorly for dynamic workloads with distributed transactions~\cite{schism}, as discussed in \S\ref{sec:background:workload-adaptiveness}.
Thus, instead of periodic repartitioning as in dynamic sharding schemes~\cite{clay, morphosys, chiller}, \system transfers ownership on demand through transaction dispatch.
This design avoids unnecessary object movement and adapts naturally to access patterns, a strategy adopted by several recent systems~\cite{hermes, zeus, leap}.
The \router maintains per-object metadata tracking the current version and leaseholder.
When scheduling a transaction to a \worker that does not own all required objects, the \router annotates the request with the current owners for each object version.
The receiving \worker fetches the necessary state directly from the owning nodes without the \router's intervention and becomes the new owner.
Upon transaction dispatch, the \router updates its metadata accordingly indicating the new \worker as the leaseholder.

\subsection{Stateless-Stateful Separation}
\label{sec:design:separation}
\system leverages a structural property of modern smart contracts: their execution can be naturally decomposed into a \emph{stateless} and a \emph{stateful} part (\textbf{Insight~1}).
The stateless part performs compute-heavy cryptographic verification and produces a validation token, which serves as an input to the stateful execution.
Before executing the stateful part, the corresponding \worker fetches this token either locally or remotely from the \worker that ran the stateless part.
If the stateless validation fails, the transaction is aborted.
To maintain version order in such cases, the \worker executing the stateful part performs a no-op that touches the relevant objects, thereby advancing their versions without applying any state changes.

This separation enables different scheduling policies for stateless and stateful parts.
The router module dispatches stateless parts with a focus on \textit{load balancing}, since they have no data dependencies.
This approach can reactively fill utilization gaps when some \workers are temporarily occupied with stateful execution due to locality-aware scheduling.
We describe how \system schedules stateful parts in \S\ref{sec:design:policy}.

\subsection{Leveraging the Consensus Window}
\label{sec:design:consensus-window}
Consensus finalization in modern blockchains often takes several hundred milliseconds (e.g., 300 ms).
Rather than leaving this interval idle, \system exploits it to hide latency and shift expensive work off the post-consensus critical path (\textbf{Insight~2}).
Specifically, during this window, \system performs stateless execution and stateful scheduling, as illustrated in Figure~\ref{fig:arch} (steps~\circled{2a} and \circled{2b}).

\noindent\textbf{Pre-consensus Stateless Execution}.
The lack of dependencies in the stateless part allows it to be executed in parallel to consensus.
This overlap leverages the consensus window to reduce end-to-end transaction execution latency.
While this can lead to wasted computation if a block is eventually rejected, such occurrences are rare, thus will have minimal impact on performance or utilization.

\noindent\textbf{Pre-consensus Stateful Scheduling}.
In addition to stateless execution, \system also runs its stateful scheduling logic during the consensus window.
The \router forwards each proposed block simultaneously to the consensus and scheduling modules.
While consensus runs, the scheduler analyzes transaction dependencies, and devises a scheduling plan.
Overlapping this with consensus allows \system to implement even more computationally expensive scheduling algorithms without affecting end-to-end latency and overall throughput.
To guarantee correctness, the router dispatches the transactions via pre-computed plan only after consensus finalizes.
On the rare occasion that consensus rejects the block, the pre-consensus stateful routing plan is invalidated and dropped.
To handle this case, \system maintains two versions of metadata: one for the scheduler and one for the router. The router updates its version when dispatching transactions.
If a block is dropped, the scheduler pauses, synchronizes the two metadata versions, discards intermediate schedules, and then resumes.

We intentionally avoid pre-consensus stateful execution.
Unlike the stateless part, stateful execution is less compute-intensive and tightly coupled with shared state dependencies.
Speculating on it would greatly complicate system design, triggering speculative object transfers and cascading aborts in the distributed setting while yielding marginal latency benefits.

\subsection{Subgraph-First Scheduling (\policy)}
\label{sec:design:policy}
\RestyleAlgo{ruled}
\begin{algorithm}[t]
    \small
    \caption{\policy (Subgraph-First Scheduling)}
    \label{alg:policy}
    \setlength{\algomargin}{1.5em}
    \SetNlSkip{0.2em}
    \DontPrintSemicolon

    \KwIn{Batch $B$ (consensus order), per-worker ownership sets $\{O_i\}$, loads $\{P_i\}$}
    \KwOut{Assignment of subgraphs to workers}

    Build dependency graph $D(B)$ and extract subgraphs $\mathcal{G}$\;

    \ForEach{subgraph $G \in \mathcal{G}$}{
      $K \leftarrow \bigcup_{t \in G} \textsc{RW}(t)$\;  \label{line:rwset}

      \ForEach{worker $i$}{
        $R_i \leftarrow |K \cap O_i| / |K|$\;  \label{line:locality}
        $L_i \leftarrow 1 - P_i / \max_j P_j$\;  \label{line:load}
        $S_i \leftarrow 0.5 * R_i + 0.5 * L_i$\;  \label{line:score}
      }

      $i^\star \leftarrow \arg\max_i S_i$ (break ties randomly)\;
      Assign $G$ to $i^\star$; update $O_{i^\star}$
    }
\end{algorithm}

Now we describe how \system efficiently schedules stateful transactions among workers.
As discussed in \S\ref{sec:background}, an effective scheduling policy must simultaneously satisfy three requirements: strict determinism, locality awareness, and load balance.
Existing approaches fall short of meeting all three at once.
For uniform workloads, a naive load balancing policy (e.g., random) would suffice.
The challenge lies in handling skewed workloads with hot objects.
Leveraging the consensus window, we design a new policy, subgraph-first scheduling (\policy), that achieves all three objectives.

Our key insight is to leverage \textbf{subgraphs} as the unit of scheduling.
For each batch of transactions from the consensus output, \system constructs a dependency graph (vertices are transactions and edges capture immediate dependencies induced by shared objects), identifies disjoint subgraphs, and dispatches those to the \workers.
The intuition is that transactions within a subgraph already share objects and ordering constraints.
Thus, collocating them on the same \worker avoids cross-node coordination.
In practice, subgraphs typically arise from contention on hot objects, making them natural candidates for collocation.

\system avoids further partitioning such subgraphs~\cite{schain}, because it would not gain performance while introducing overhead.
For extremely contended workloads, distributed stateful execution yields no benefits since the maximum concurrency is limited by the workload itself.
Therefore, the contribution of this scheduling scheme is on low to medium level of contention.
This idea aligns with prior work in distributed transactions: skewed, contended workloads benefit from single-node execution, while uniform workloads distribute effectively across nodes~\cite{turbodb}.

Algorithm~\ref{alg:policy} summarizes \policy.
For each consensus batch, the \router builds the dependency graph and schedules each disconnected subgraph $G$ as a unit.
It first computes the subgraph’s object footprint $K$ as the union of the read/write sets $\textsc{RW}(t)$ of all transactions $t \in G$.
Then, for each \worker $i$, \policy computes a \emph{locality} score $R_i$ as the fraction of $K$ already owned by $i$ (using the ownership map $O_i$), and a \emph{load} score $L_i$ from $i$’s current load $P_i$.
The final scheduling score $S_i$ is the equal-weighted combination of locality and load (line~\ref{line:score}).
This jointly optimizes both, i.e., locality score reduces state movement and cross-\worker coordination, while load score prevents persistent skew and state accumulation on a single \worker~\cite{leap, hermes}.
\policy assigns the entire subgraph to the \worker with the highest score and updates the ownership metadata accordingly.
Stateless work is dispatched purely by load, which helps utilize idle \workers even when locality concentrates some stateful subgraphs.

\subsection{Handling Failures}
\label{sec:design:ft}
\system scales out single-validator execution and adopts the same failure model as existing blockchains that allows for validator crashes.
In \system, we treat the failure of \router, the sole externally visible interface to consensus, as equivalent to a single-validator failure.
In that case, the remaining validators continue to make progress as long as the number of failed validators remains within the fault threshold of the underlying consensus protocol.
To avoid reducing the validator's MTTF due to the distributed architecture, we carefully design \system to tolerate \worker failures.

\noindent\textbf{Periodic snapshotting.}
\system implements \textit{periodic snapshotting} to keep the \router loosely but consistently synchronized with the execution state in the \workers.
Snapshotting operates in \textit{epochs}.
The \router has full flexibility in determining the epoch size, choosing boundaries that either align with consensus rounds or adapt to workload characteristics, and explicitly notifies the \workers when an epoch ends.
This exposes a tunable trade-off: longer epochs reduce snapshot overhead but increase recovery cost, while shorter epochs invert this trade-off.
\workers keep track of the objects they modify within an epoch and once the epoch finishes, they send the \router an \emph{incremental} update containing only the latest versions of those modified objects.
This incremental checkpointing on \workers amortizes frequent updates to hot objects without communicating every versioned update.
With its global view of the system, the \router knows exactly which object versions it needs at the end of each epoch.
After collecting updates from all \workers, it atomically updates its state and advances \texttt{persist\_index}, which marks the log position up to which all preceding transaction effects have been durably persisted.

\noindent\textbf{Failure recovery.}
When a \worker fails, the \router initiates recovery as follows.
It first spawns a new \worker and determines the transaction \textit{replay set}, defined as the range between the current dispatch index and the \texttt{persist\_index}.
From the log view, state up to \texttt{persist\_index} is durable at the \router.
Thus, replaying after that point is required to reconstruct the parts of the missing suffix.
The \router reclaims ownership of all objects previously assigned to the failed \worker that are already durably persisted, and then dispatches the entire replay set to the new \worker.
During replay, the \router attaches required state that it has durably persisted, avoiding extra round-trip fetches.
This procedure brings the new \worker online lazily, driven by transaction demand, to reduce stalls of excessive state transfers.
After dispatching the replay set, the \router can continue processing incoming transactions.
We conservatively replay the entire replay set rather than attempting selective replay, as transactions may have been partially executed by either healthy or failed \workers and failures can occur at arbitrary points during execution.
Notably, our design allows for skipping re-running all stateless parts as long as the \router receives acknowledgement upon successful validation, which we leave for future work.


\subsection{Elasticity}
\label{sec:design:elasticity}
\system supports seamless elasticity by automatically scaling the \worker pool in response to workload fluctuations.
The \router orchestrates resource allocation by periodically monitoring per-\worker load and deciding when to spawn or retire \workers.
Crucially, autoscaling should not stall foreground execution or introduce any explicit rebalancing phase: \system leverages its lease-based ownership model to migrate objects across \workers and \router dynamically and lazily.

\noindent\textbf{Scale-out.}
The \router spawns a new \worker and begins dispatching transactions to it.
Initially, the new \worker has no local state.
As it executes transactions, it fetches any objects it does not own from other \workers or from the \router.
Over time, state naturally migrates to the new \worker as leases are transferred during normal execution.
Note that a new \worker can immediately start helping with stateless execution.

\noindent\textbf{Scale-in.}
The \router retires a \worker \emph{gracefully} in two steps.
First, it stops assigning new transactions to the retiring \worker, but keeps it online for a grace period to serve in-flight requests and inter-\worker lease transfers.
Second, to complete retirement and reclaim ownership atomically, the \router terminates the current snapshotting epoch for that \worker, forcing it to report all objects modified in that epoch back to the \router.
This yields an epoch-aligned handoff:
hot objects migrate away naturally as their leases are transferred to other \workers, while cold objects are reclaimed by the \router via the epoch snapshot where future transactions fetch them directly from the \router.

\section{Implementation}
\label{sec:impl}

We implemented \system in 13k LoC of Rust.
This includes a modular implementation of the \router and \workers.
We use tokio~\cite{tokio} for asynchronous network IO across nodes which communicate over TCP sockets.

\noindent\textbf{\router}.
Each module runs as a long-lived task pinned to a dedicated core.
Modules on the critical transaction path (Figure~\ref{fig:arch}) are connected by bounded channels, forming a pipeline.
The router module spawns lightweight dispatching tasks to forward transactions according to the scheduling policy for stateful and stateless parts.
These tasks run to completion in parallel on a thread pool, occupying all remaining cores, with one thread pinned per core.
The autoscaler module periodically monitors the incoming load rate.
\router detects \worker failures by observing lost network connections and triggers the failure handling scheme (\S\ref{sec:design:ft}).

The \router maintains metadata for each object and each node in the cluster.
The object metadata takes the form of a map from object names to a tuple of \texttt{(current\_version, current\_owner)}.
The node metadata maintains per-\worker load information, which the stateless and stateful routing leverage to make load-aware decisions.
The metadata footprint is modest:
for 10M objects, the size of ownership metadata totals around 100\,MB.
This is in line with production chains; for example, although hundreds of millions of addresses have appeared on Ethereum, daily active addresses peak around 1.4M~\cite{etherscan_active_addresses}.

\noindent\textbf{\worker}.
Each \worker node in \system runs a deterministic parallel runtime that handles both stateless and stateful compute.
The runtime implements a custom asynchronous thread pool to drive the task execution.
The design is agnostic to the underlying smart contract VM and focuses on parallel execution with strict ordering guarantees.
Each execution unit, stateless or stateful, is an asynchronous task.
Stateless tasks are immediately runnable and can run in parallel.
Stateful tasks depend on (i) their corresponding stateless tasks (e.g., authentication) and (ii) any previous tasks accessing the same objects.
The runtime uses a dynamic DAG to keep track of the dependencies among tasks based on the specific object versions they access as assigned by the \router.
A certain task can only be executed when all its prior dependencies are satisfied.

The intra-\worker runtime implements the task DAG via \texttt{tokio::\allowbreak{}sync::\allowbreak{}Notify}~\cite{tokio-notify}, a lightweight primitive for synchronization.
\texttt{Notify} carries no payload and is used purely for readiness signaling.
Concretely, the runtime maintains one \texttt{Notify} instance per object version.
When a transaction arrives with its required versions explicitly annotated, the runtime spawns an asynchronous task.
This task (1) \texttt{await}s on the \texttt{Notify} instances for all required input versions, ensuring all dependencies are satisfied;
(2) once the dependencies become available executes the transaction to completion on the thread pool;
(3) signals the \texttt{Notify} instances corresponding to the versions it produces, thereby unblocking dependent tasks.
Remote dependencies work similarly, as when the \worker fetches the object, it signals the equivalent \texttt{Notify} instance.
Finally, since each version is consumed exactly once, the runtime garbage-collects the corresponding \texttt{Notify} entry and the object version immediately after it has been consumed.
Each \texttt{Notify} object has a unique name within the \worker and the runtime creates it atomically the first time it processes the producer or the consumer task.
This design naturally integrates dependency enforcement into the asynchronous runtime, enabling deterministic parallel execution without any invasive changes to the existing scheduler, i.e., Tokio's in the current implementation.

\section{Evaluation}\label{sec:eval}
Our evaluation aims to answer the following questions:
\begin{itemize}[noitemsep,topsep=3pt,leftmargin=*]
	\item Does \system's \textit{asymmetric} architecture improve efficiency? (\S\ref{sec:eval-arch})
	\item How does \policy compare against other distributed (non-)deterministic scheduling schemes? (\S\ref{sec:eval-policy})
	\item How much does the stateless-stateful separation improve performance? (\S\ref{sec:eval-separation})
	\item What is the benefit of leveraging the consensus window? (\S\ref{sec:eval-pre-consensus})
	\item How efficiently does \system recover from \worker failures? (\S\ref{sec:eval-failure})
	\item How does \system autoscale under varying loads? (\S\ref{sec:eval-elastic})
	\item Can \system handle realistic dynamic workload hotspots? (\S\ref{sec:eval-dynamic-hotspot})
	\item How scalable is \system's \router? (\S\ref{sec:eval-scalability})
\end{itemize}

\subsection{Experimental Setup}
\noindent\textbf{Testbed.}
We run our experiments on a 16-node cluster on AWS, using \texttt{m5d.8xlarge} instances in the same region.
One node serves as the client node which generates transactions based on an open-loop Poisson process, one node serves as the \router, and the others as \workers.
Each machine is equipped with a 10Gbps NIC, 32 vCPUs (16 physical cores), and 128GB memory, which aligns with the recommended hardware configuration for typical validators~\cite{sui-min-specs, solana-min-specs, aptos-min-specs}.
We follow the layered approach of building blockchains~\cite{PoA}.
\system is an execution layer system, hence our experiments focus on the execution layer of a single validator, which is the component \system changes.
For the consensus layer, we implement a mock consensus module that adds a constant delay of $300ms$ equivalent to the consensus execution~\cite{mysticeti}.
\system is agnostic to consensus deployment: variations such as co-located vs. geo-distributed validators, heterogeneous capacities, or Byzantine behavior would manifest only as variability in the consensus window duration.
Shorter windows would reduce the slack for pre-consensus optimizations and longer windows would strictly increase these benefits.

\noindent\textbf{Workloads.}
Following the methodology of prior work~\cite{spectrum, hermes, delaypart}, we evaluate \system using YCSB~\cite{ycsb}, TPC-C~\cite{tpcc}, and Ethereum mainnet traces~\cite{pevm}.
For YCSB, we implement a smart contract with a global mapping that represents a key-value store over a 10M-key keyspace, where each object is modeled as a 0.25\,kB account.
To study contention and skew, we vary the number of keys accessed per transaction under both uniform and Zipfian distributions.
Each transaction contains a fixed number of read-modify-write operations, each on a unique key.
Since \system is VM-agnostic, we factor out VM-specific effects by modeling execution as a stateless stage followed by a stateful stage.
Because real smart-contract VMs incur high overhead while pure YCSB reads/writes are lightweight, we add synthetic spinning to emulate realistic execution.
Unless otherwise stated, YCSB uses 0.5\,ms synthetic time for both stages, consistent with \S\ref{sec:background} and prior measurements on common smart contracts~\cite{crystality}.
We also validate real execution by running the same experiments on SuiVM and real-world cryptographic functions (\S\ref{sec:eval-separation}).
For TPC-C, we implement a smart contract supporting New-Order and Payment transactions, which account for 88\% of the workload and involve remote data access~\cite{leap}, with 40 warehouses per \worker.
Since TPC-C is already heavier than YCSB, its stateful stage runs the real business logic, while the stateless stage uses the same 0.5\,ms synthetic cost.

\subsection{Architectural Efficiency}\label{sec:eval-arch}

\begin{figure}
	\centering
	\includegraphics[width=1\columnwidth]{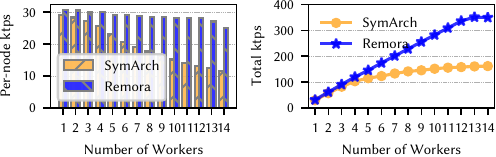}
	\caption{Comparison with other work using symmetric architecture.}
	\label{fig:eval-arch-diff}
\end{figure}

\noindent\textbf{Methodology.}
We start by comparing \system's \emph{asymmetric} architecture with a single node in charge of scheduling to prior Calvin-inspired~\cite{calvin} \emph{symmetric} schemes~\cite{hermes, delaypart, qstore, slog, detock} where every node accepts and deterministically schedules all transactions.
To evaluate the efficiency of these two choices, we implement a baseline, \texttt{SymArch}, where each \worker performs both deterministic scheduling and execution to represent this line of work.
Since sequencing in our setting is established via consensus, we exclude sequencing from both systems to ensure a fair comparison.
We use the YCSB benchmark with uniform distribution and a round-robin scheduling policy instead of \policy, while disabling the stateless-stateful separation.
In this configuration, the ideal per-\worker throughput is about 32k TPS, given 32 vCPUs and a 1 ms per-transaction service time.
We vary the number of \workers and measure the maximum throughput.

\noindent\textbf{Results.}
Figure~\ref{fig:eval-arch-diff} shows that \system delivers higher throughput per deployed node than \texttt{SymArch} as the number of \workers increases.
In \texttt{SymArch}, every \worker repeatedly pays the cost of deterministic scheduling, consuming CPU cycles that could otherwise execute transactions.
This redundant overhead directly reduces the useful work obtained from each provisioned machine.
In contrast, \system offloads all scheduling to the \router, allowing \workers to dedicate their full capacity to transaction execution, improving resource efficiency and translating the same hardware budget into more throughput.
Although \system requires an additional coordination node, the net deployment cost is still lower for a target throughput: for example, a 4-worker \system configuration already outperforms a 5-worker \texttt{SymArch}.
We expect the efficiency gap to widen further with more advanced scheduling policies that are more computationally expensive compared to round robin.
\system's near-linear scaling stops once the \router becomes the bottleneck around 340 kTPS, which we explain in \S\ref{sec:eval-scalability}.
Notably, \system can also support more \workers when smart contract is compute-heavier, or when \worker instance is smaller.

\subsection{Scheduling Schemes Performance}\label{sec:eval-policy}

\begin{figure*}
	\centering
	\includegraphics[width=2\columnwidth]{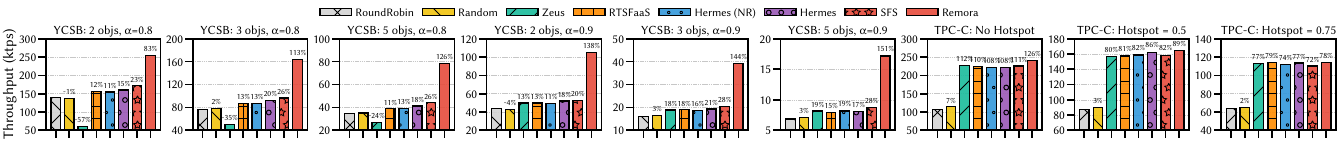}

	\caption{Performance comparison with other policies. Percentage denotes the improvement over round-robin policy.}
	\vspace{-0.4cm}
	\label{fig:eval-policy}
\end{figure*}

\begin{figure*}
	\centering
	\includegraphics[width=2\columnwidth]{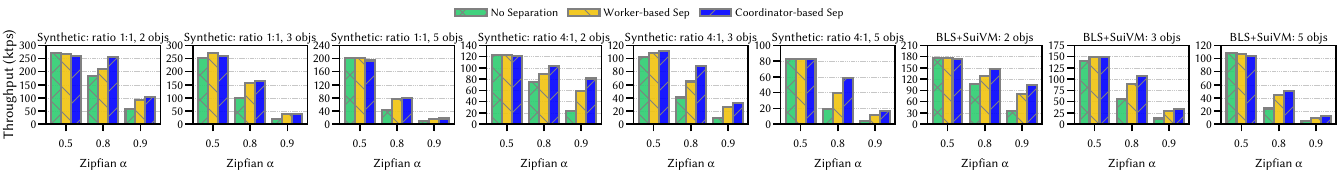}
	\caption{{Benefits of stateless-stateful separation.
	Ratio denotes the duration ratio of stateless and stateful part in each transaction.}
	}
	\label{fig:eval-separation}
\end{figure*}

\noindent\textbf{Methodology}.
{\sloppy
We compare \policy against several state-of-the-art scheduling schemes, including both deterministic and non\-deterministic designs by measuring the max achieved throughput.
\textbf{Hermes}~\cite{hermes} represents the state-of-the-art for deterministic scheduling.
It considers both load and locality but violates strict determinism, since it relies on transaction reordering to improve performance.
We also implement a strictly-deterministic variant, \textbf{Hermes (NR)}, that disables reordering in Hermes.
\textbf{RTSFaaS}~\cite{rtsfaas} is another scheme that considers both locality and load balance initially designed for FaaS workloads, yet is not deterministic.
\textbf{Zeus}~\cite{zeus} is a locality-only baseline that aggressively migrates objects to the \worker executing each distributed transaction.
We also include simple load-balancing policies: RoundRobin and Random.
Notably, under uniform or low-skew workloads, all policies achieve similar performance as transactions are evenly distributed and require minimal object migration.
This experiment thus focuses on contended workloads.
For an apples-to-apples comparison of scheduling policies, we disable \system's use of the consensus window and the stateless-stateful separation.
\par}

\noindent\textbf{YCSB}.
As shown in Figure~\ref{fig:eval-policy}, naive policies perform the worst due to frequent distributed transactions and object transfers.
Zeus ignores load imbalance and gradually migrates all objects to a single \worker.
After warm-up, only one \worker is busy while others remain idle, which is consistent with prior observations~\cite{hermes}.
However, under high contention (3 or 5 objects with $\alpha=0.9$) in which the concurrency level in the workload is extremely limited, the benefits of scale-out diminish and single-node execution outperforms others.
RTSFaaS and Hermes (NR) achieve similar performance ($11$\% to $19$\% improvement over round-robin) by considering both load and locality.
The full Hermes policy benefits further from reordering up to $21$\%, yielding the best baseline performance.

\policy consistently outperforms all baselines (up to $28$\%), including non-deterministic ones, across all contended configurations.
Whereas Hermes and RTSFaaS balance load by migrating individual transactions, \policy schedules at the subgraph level, colocating dependent and locality-sharing transactions.
This granularity is dependency-aware and minimizes remote state transfer, confirming that subgraphs are the right unit for scale-out scheduling.
We also include \system, which is the \policy policy with the optimizations of stateless-stateful separation and consensus window execution.

\noindent\textbf{TPCC}.
For TPC-C benchmarks, we pre-partition warehouses among \workers, and construct the hotspot by setting 50\% and 75\% of transactions targeting at the warehouses pre-assigned on the same \worker.
\policy does not yield noticeable throughput gains over the baselines on TPC-C.
This is expected because TPC-C uses warehouse-based partitioning by design, and most transaction state accesses are by the home warehouse/district while cross-warehouse accesses are relatively limited (default 15\% in NewOrder and Payment transactions).
Consequently, under TPC-C the bottleneck is dominated by warehouse-level load skew rather than distributed execution, so different locality-aware policies converge in performance.
Our observation is consistent with prior work~\cite{hermes, delaypart}.

\noindent\textbf{Lease transfer analysis.}
We measure the network consumption of \workers involved in object migration and report their overhead.
Under uniform YCSB cases where total throughput reaches 100\,k TPS, lease transfer with two objects per transaction consumes 24\,MiB/s of network bandwidth, of which 73\% is inbound due to transactions received from \router.
In contended cases where $\alpha=0.9$, the bandwidth consumed by the most heavily loaded \worker{} increases to 33\,MiB/s.
Bandwidth consumption also rises with the number of objects per transaction.
For example, with five objects per transaction, the consumed bandwidth increases to 62\,MiB/s.
Overall, this level of network consumption does not pose a deployment concern for modern cloud infrastructures.

\subsection{Impact of Stateless-Stateful Separation}\label{sec:eval-separation}

\noindent\textbf{Methodology.}
We study the benefits of separating the stateless from the stateful part of a transaction in the achieved throughput.
We consider two separation modes:
(i) \textsl{Worker-based Sep}, where only \workers apply stateless–stateful separation during execution.
This configuration could even be applied in single-node runtimes or distributed runtimes without centralized scheduling.
(ii) \textsl{Coordinator-based Sep}, in which the \router can schedule stateless and stateful parts on different \workers.
We use the \policy scheduling policy and the YCSB benchmark to study such impact.

\noindent\textbf{Results.}
Figure~\ref{fig:eval-separation} shows that \textsl{Worker-based Sep} improves throughput by up to $2\times$, as separating execution at \workers unlocks greater parallelism across available cores.
\textsl{Coordinator-based Sep} achieves up to $3\times$ improvement, with the most pronounced advantage when the stateless dominates.
These results highlight that under skewed workloads, separation at the \router effectively harnesses otherwise idle compute capacity across \workers.
In contrast, under low contention the benefits diminish, as load is balanced.

\subsection{Benefits of Consensus Window}\label{sec:eval-pre-consensus}

\begin{figure}
	\centering
	\includegraphics[width=1\columnwidth]{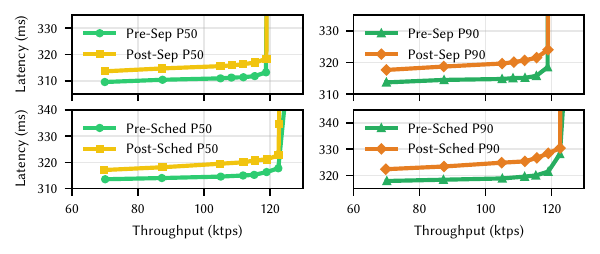}
	\caption{Benefits of consensus window. \textsl{Pre-Sep} denotes pre-consensus stateless execution, while \textsl{Post-Sep} denotes post-consensus stateless execution; \textsl{Pre-Sched} denotes pre-consensus scheduling, while \textsl{Post-Sched} denotes post-consensus scheduling.
	}
	\label{fig:eval-consensus-window}
\end{figure}

\noindent\textbf{Methodology.}
\system leverages the consensus window to perform stateless execution and stateful scheduling.
We study the effect of this design choice on latency.
Given that the benefits are independent of scheduling schemes, we use the uniform YCSB benchmark, 2ms of stateless, and 0.5ms of stateful processing, and \policy scheduling.
As a baseline, we disable pre-consensus scheduling and decoupled stateless execution and measure the end-to-end latency with achieved throughput.

\noindent\textbf{Pre-consensus stateless execution.}
In Figure~\ref{fig:eval-consensus-window}, \textsl{Pre-Sep} demonstrates clear latency improvement by overlapping stateless execution with the consensus window.
\textsl{Post-Sep} incurs an additional 2\textasciitilde 5ms at both P50 and P90 latency.
Given that the stateless portion must always be executed and is fully parallelizable, the performance gain from pre-consensus execution is guaranteed, and becomes even more pronounced as the stateless workload grows.

\noindent\textbf{Pre-consensus stateful scheduling.}
A similar effect arises when scheduling is shifted into the consensus window.
To isolate this effect, we keep stateless execution in post-consensus for this experiment.
Although our scheduling policy is lightweight (\textasciitilde 4ms per batch), latency reduction of \textsl{Pre-Sched} indicates that the consensus window provides sufficient slack to accommodate more sophisticated scheduling strategies off-path, e.g., ML-based ones.


\subsection{Failure Recovery}\label{sec:eval-failure}

\begin{figure}
  \centering
  \includegraphics[width=1\columnwidth]{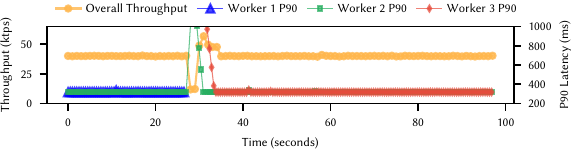}
  \caption{Failure recovery. \worker 1 failed and \worker 3 is spawned.
  }
  \label{fig:eval-failure}
\end{figure}

We evaluate \system's failure recovery by subjecting it to a constant YCSB uniform workload while injecting a \worker failure at a random time.
From this experiment onward, we use a smaller 3-\worker pool, as these experiments focus on demonstrating \system's mechanisms and the feasibility of properties rather than its performance.

Figure~\ref{fig:eval-failure} illustrates the recovery process.
Initially, \worker 1 and \worker 2 jointly handle the incoming requests.
Upon detecting the failure of \worker 1 at ${t=28s}$, \system immediately switches on \worker 3 as a replacement.
We omit the cost of provisioning an AWS EC2 VM and instead pre-launch an idle \worker before the experiment.
During recovery, the system experiences a brief throughput degradation of approximately 2 seconds while it determines the replay set, transfers object ownership, and dispatches the replay set to \worker 3.
Subsequently, \worker 2 and \worker 3 enter a catch-up phase, processing requests at their maximum throughput capacity, which temporarily elevates latency.
Once the backlog is cleared, the performance stabilizes at the pre-failure levels, demonstrating \system's resilience to \worker failures.

\subsection{Elasticity}\label{sec:eval-elastic}

\begin{figure}
	\centering
	\includegraphics[width=1\columnwidth]{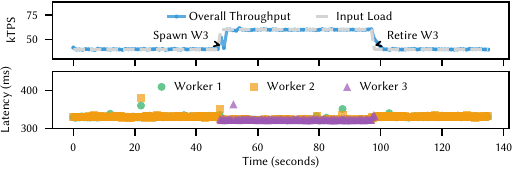}
  	\caption{Elastic scaling for load fluctuations.} 
	\label{fig:eval-elasticity}
\end{figure}
We study \system's elasticity capabilities by applying load that changes over time via uniform YCSB workloads.
Figure~\ref{fig:eval-elasticity} shows the input load pattern: throughput increases sharply from 40k TPS to 60k TPS and then decreases to 40k TPS.
Once the autoscaler detects the load spike which exceeds the processing capacity of existing \workers, it starts a new \worker instance (pre-launched as same as the failure experiment).
The \router immediately begins forwarding load to the new \worker, which incrementally fetches required state from existing nodes on demand.
When the load drops to 40k TPS, the autoscaler initiates scale-in and gracefully retires \worker~3.
Across both transitions, latency remains low and stable, showing that \system can autoscale without stalling foreground execution or triggering an explicit rebalancing phase.

\subsection{Handling Dynamic Hotspots}\label{sec:eval-dynamic-hotspot}
\begin{figure}
	\centering
	\includegraphics[width=1\columnwidth]{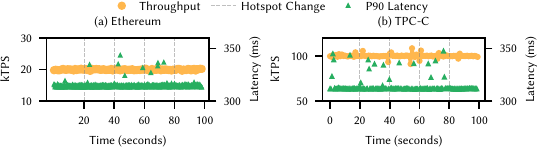}
  	\caption{Dynamic hotspot workloads.} 
	\label{fig:eval-dynamic-hotspot}
\end{figure}

We evaluate \system's ability to handle dynamic hotspots, i.e., object popularity changes, via both real-world traces from the Ethereum mainnet~\cite{pevm} and TPC-C benchmarks.
For the Ethereum traces, we derive object access distributions from different blocks spanning different periods of time to capture realistic contention.
We run this workload on a production-grade smart contract VM, SuiVM~\cite{sui-code}.
For the TPC-C experiment, we target 50\% of the transactions to warehouses on a single \worker node to emulate hotspots.
The workload changes every 20s, causing hotspots to shift dynamically across \workers and mimicking temporal locality observed in practice.
We set the input rate to 80\% of the measured system capacity for each workload (accounting for available parallelism and VM overhead) based on offline profiling.
\system shows only occasional high-latency outliers (Figure~\ref{fig:eval-dynamic-hotspot}).

\subsection{Scalability}\label{sec:eval-scalability}
\begin{figure}
  \centering
  \begin{minipage}[c]{0.32\linewidth}
    \centering
    \includegraphics[width=\linewidth]{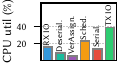}
    \label{fig:cpu-breakdown}
  \end{minipage}\hfill
  \begin{minipage}[c]{0.66\linewidth}
    \centering
    \label{tbl:scalability}
    \footnotesize
    \begin{tabular}{c|ccccc}
      \toprule[1.5pt]
      Policy &
      \begin{tabular}[c]{@{}c@{}}Round\\robin\end{tabular} & Zeus &
      \begin{tabular}[c]{@{}c@{}}RTS\\FaaS\end{tabular} & \policy & \system \\
      \midrule
      TPS & 380k & 360k & 354k & 348k & 251k \\
      \bottomrule[1.5pt]
    \end{tabular}
  \end{minipage}
  \vspace{-0.3cm}
\captionof{figure}{\router analysis. The left figure shows the CPU utilization breakdown, while the right table shows the max throughput.}
  \label{tbl:scalability}
\end{figure}

While \system's asymmetric architecture enables near-linear scale-out with the number of \workers (\S\ref{sec:eval-arch}), centralizing transaction dispatch at the \router raises the concern that it could become a throughput bottleneck.
We therefore stress-test the \router under a range of scheduling policies.
To isolate dispatching overhead, we use uniform distribution with zero service time, ensuring that \workers do not limit throughput.
As shown in Figure~\ref{tbl:scalability}, with a simple round-robin policy the system sustains nearly 380k TPS.
Adding computation for locality calculation and dependency analysis reduces throughput only slightly, to around 350k TPS.
Separating stateless and stateful execution reduces throughput to 250k TPS due to the extra message serialization, I/O, and syscall overhead.
Given that state-of-the-art consensus modules achieve 200\textasciitilde 300k TPS~\cite{mysticeti, shoal++}, we conclude that the \router in our design is sufficient for practical deployment.

\begin{figure}
  \centering
  \includegraphics[width=1\columnwidth]{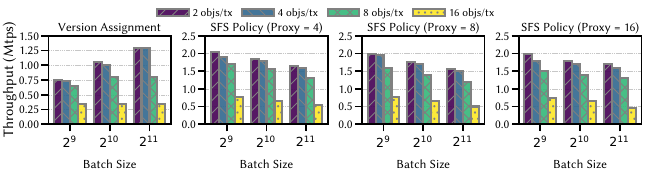}
  \caption{VerAssign. and Sched. (\policy) module sensitivity analysis.
  }
  \label{fig:eval-submodule}
\end{figure}

To dive deeper into the \router's design, Figure~\ref{tbl:scalability} also reports the per-module CPU utilization breakdown per 1k transactions.
Notably, when removing stateless-stateful separation, the TX-IO part only takes 25\%, demonstrating the extra overhead of the separation.
Leveraging syscall batching and kernel bypass could mitigate this, which we leave for future work.
We also observe that all cores are saturated, indicating that larger EC2 instances would yield proportionally higher throughput, as the network tasks are parallelizable.

In all these modules, version assignment is the main \textit{serial} component, making it the primary scalability bottleneck as per Amdahl's Law.
We also study each module in isolation to identify which parameters affect performance.
We find that the number of objects per transaction is the key limiting factor for version assignment.
In the scheduling module, though building dependency graphs can be parallelized, the locality-related part is sequential as to know the scheduling decision made for overlapped objects.
As shown in Figure~\ref{fig:eval-submodule}, we sweep the number of \workers, batch size, and objects per transaction, and find that \policy's throughput is mainly sensitive to the latter two, due to dependency-graph construction and subgraph assignment overhead.

\section{Discussion}\label{sec:discussion}

\noindent\underline{\textbf{Read/write set assumption}}.
\system relies on knowing each transaction's read/write set before stateful execution.
This requirement restricts expressiveness: it does not directly support contracts whose accessed state is discovered only during execution.
Thus, \system targets execution environments where access sets are part of the transaction interface, can be conservatively derived before execution, or are provided by clients/builders through off-chain simulation.
Such environments still cover many common smart-contract workloads, including token transfers, decentralized exchanges, and NFT mints.
This model is common in deterministic database systems~\cite{calvin, bohm, pwv, slog, detock, caracal, qstore, caerus, doradd} and research blockchain systems~\cite{peep, schain, parblockchain}, and has been shown to hold for a majority of transactional workloads in a recent workload-characterization study~\cite{rw-assumption}.
Some prior systems based on optimistic execution~\cite{aria, harmony, occ-da} relax this limitation, but at the cost of violating strict determinism.
Recent work~\cite{dag-btm, nemo} further shows that access specifications can benefit optimistic execution by reducing aborts.
This constraint simplifies our enforcement of strict determinism and enables efficient scheduling with locality awareness.
Furthermore, it is also a practical design point in several production blockchains: Solana and Sui expose transaction-level account/object access metadata for parallel execution~\cite{solana-rw, sui-rw}, while Ethereum is set to enforce block-level access lists~\cite{eip7928} as a headliner feature in its upcoming Glamsterdam upgrade~\cite{eth-bal}.
Thus, \system can be directly integrated into any existing blockchain system with a modular execution layer whose access set can be determined before execution.

\noindent\underline{\textbf{Scaling L1 on-chain compute}}.
Emerging on-chain workloads (AI inference, ZK proofs, and post-quantum crypto~\cite{algorandPostQuantum2025, dfinity2024deai, wang2025aiarena, fastcrypto}), are compute-heavy and currently pushed to L2s or off-chain due to incapability of current L1 infrastructure.
\system removes this reliance via elastic scale-out, improving L1 performance while avoiding fragmentation and enabling new L1-native use cases.

\noindent\underline{\textbf{DoS resilience}}.
Stateless validation often involves expensive cryptographic checks, making it a common vector for DoS attacks.
\system mitigates this risk by elastically scaling stateless execution across nodes, absorbing heavy or adversarial workloads without introducing coordination bottlenecks.
This elasticity ensures system responsiveness and robustness even under targeted load.
As blockchains adopt quantum-safe cryptography and other compute-intensive primitives, such resilience becomes increasingly essential.

\section{Related Work}\label{sec:relatedwork}


\noindent\textbf{Deterministic Distributed Transactions}.
Most scale-out deterministic transactional systems adopt Calvin~\cite{calvin}'s \emph{symmetric} architecture, unlike \system's asymmetric design.
Some focus on improving the sequencing layer, while others optimize the partitioning and scheduling scheme.
SLOG~\cite{slog} and Detock~\cite{detock} decentralize the sequencer to support geo-distributed deployments.
Caerus~\cite{caerus} minimizes the round-trip traversal by deterministically merging per-region partial sequences using a new ordering protocol.
Q-Store~\cite{qstore} bypasses the single-threaded sequencing and scheduling parts via constructing fine-grained execution queues.

\noindent\textbf{Deterministic Single-Node Transactions}.
Deterministic parallelism has been widely studied in databases and systems.
A broad line of work follows a pessimistic \emph{schedule-then-execute} pattern, with mechanisms spanning dependency-graph analysis~\cite{aria, doradd, kuafu, cbase, pwv}, multi-versioning~\cite{bohm, caracal}, and per-object queues~\cite{decentsched}.
\system similarly decouples scheduling from execution.
Other work adopts optimistic \emph{execute-then-validate} schemes~\cite{aria, eve}, which avoid pre-declared read/write sets but can incur aborts under contention and relax strict adherence to the agreed-upon order.

\noindent\textbf{Non-Deterministic Distributed Transactions}.
This line of work relies on variants of two-phase commits (2PC) to coordinate multi-partition transactions.
Some accelerate 2PC using new hardware technologies such as RDMA~\cite{redt, chardonnay} and CXL~\cite{tigon, pasha} to reduce communication costs.
Eris~\cite{eris} exploits programmable switches for transaction sequencing.
Centralized transaction routing with ownership-based data migration is also adopted in several shared-storage databases~\cite{chimera, gaussdb}.
GaussDB~\cite{gaussdb} employs a centralized routing model performing inference while deploying training elsewhere to avoid the on-path overhead.
Our proposed consensus window can allow such computationally expensive and ML-based routing schemes.
Pegasus~\cite{pegasus} uses an in-network directory to track objects location, similar to how \router tracks ownership.
Lotus~\cite{lotus} optimizes multi-partition transactions by introducing granule locks to mitigate distributed commit overhead.
Besides, epoch-based commit/recovery and asynchronous persistence in prior work~\cite{coco, persimmon} are similar to our design.

\noindent\textbf{Smart Contract Execution}.
Existing efforts primarily exploit multi-core parallelism for deterministic execution.
Recent work targets finer-grained operation-level optimizations to mitigate rollback overhead under optimistic concurrency control~\cite{parallelevm,spectrum,occ-da,nezha,tell}.
OptME~\cite{optme} builds key-based dependency graphs to derive parallel execution schedules while extensively using re-ordering.
Some other work leverages new programming semantics to enhance performance, including commutativity~\cite{crystality, groundhog} and deferred objects~\cite{deferred-objects}.
Spectrum~\cite{spectrum} carefully designs rollback and re-execution mechanisms to support unknown read/write sets while preserving strict determinism.
Pre-execution is also used to synthesize accelerated EVM code by caching speculative results for all possible branches~\cite{forerunner,seer}.
DeCl~\cite{decl} adopts software-based sandboxing to perform deterministic gas metering on untrusted machine code.
In Vegeta~\cite{vegeta}, each validator speculatively executes its proposed blocks to generate hints, and after consensus all validators replay all blocks.
These mechanisms are orthogonal to the design of \system.

\noindent\textbf{Blockchain Architectures}.
AdaptChain~\cite{adaptchain} scales throughput by adjusting concurrent block generation based on demand.
Some systems move state off-chain to increase throughput~\cite{slimchain,porygon}.
Pilotfish~\cite{pilotfish} adopts the Calvin architecture and uses fixed data partitioning, thus inelastic and inadaptive to workloads.
In permissioned systems, SChain~\cite{schain} adopts a similar architecture in its single trusted organization and decouples execution into compute and storage roles.
FlexChain~\cite{flexchain} leverages hardware disaggregation to improve resource utilization.
\system differs in leveraging domain insights and focusing on the performance and operational practicality of a modular execution layer, instead of the whole chain.

\section{Conclusion}\label{sec:conclusion}
In this work, we make the case for a scale-out blockchain execution layer that preserves strict determinism.
We show why prior deterministic transactional systems fall short and identify opportunities unique to blockchains.
Based on these insights, we design \system, a scale-out smart contract execution engine.
\system adopts an asymmetric architecture for efficiency and cost-effectiveness, and employs versioning with ownership to guarantee determinism.
\system separates the stateless part from smart contract transactions and leverages consensus window to perform stateless execution and stateful scheduling.
\system also demonstrates its elasticity and workload adaptiveness.
We believe \system paves the way for the next generation of scalable blockchain systems.

\begin{acks}
We sincerely thank anonymous reviewers, George Danezis, Kostas Chalkias, Ben Livshits, and the members of the LSDS group for their feedback.
This work is partially supported by Mysten Labs and a Sui Academic Research Award.
\end{acks}

\bibliographystyle{ACM-Reference-Format}
\bibliography{ref}

\clearpage
\appendix
\section*{APPENDIX}
\label{appendix}
\section{System Summary}\label{sec:system-model}

\subsection{Entities}

A validator in \system consists of two types of entities that work together to execute transactions. The \textbf{\router} serves as the validator's centralized coordinator, performing the following functions: it interfaces with the consensus layer, assigns per-key versions in input order, schedules and dispatches transactions to workers, tracks object ownership through leases, controls elasticity by managing the worker pool, persists snapshots for durability, and coordinates recovery after failures.

The \textbf{\workers} are a dynamic set of execution nodes responsible for transaction processing. These nodes execute both stateless and stateful parts of transactions, maintain leased objects in memory, fetch objects and accept lease transfers on demand from other workers, and perform local readiness gating to determine when transactions can execute. Unlike the \router, workers do not engage in global scheduling decisions, focusing instead on local execution.

Additionally, a \consensus module runs on the \router, producing committed blocks that establish a total order of transactions. A \emph{block} is an ordered batch of transactions that are processed atomically by the consensus layer. The \router can also observe blocks proposed to consensus before their commitment, enabling speculative execution.

\subsection{State and Objects}
The global state in \system is represented as a versioned map $(\mathcal{K}\times\mathbb{N})\to\mathcal{V}$, where each pair $(k,i)$ identifies the value of object $k$ at version index $i$. 

For each key $k$, the \router maintains a monotonically increasing counter $\mathrm{ver}[k]\in\mathbb{N}$ that tracks the next version index to be assigned. An \emph{object version} is defined as a pair $(k,i)$ where $0\leq i\leq\mathrm{ver}[k]$. The system assumes an initial state where every key $k$ has an initial version $(k,0)$ representing its initial value.

\subsection{Transactions and Separation}\label{sec:transactions-separation}
Each transaction $T$ declares static read/write sets $R(T), W(T)\subseteq\mathcal{K}$ that are known to the \router before scheduling begins. This static declaration enables the \router to perform dependency analysis and version assignment without executing the transaction.

Transaction execution follows a two-phase separation model. The \textbf{stateless step} consists of a validation function $\mathrm{Validate}(T)$ that returns either a token $\tau_T$ or an invalid result. This token is bound to the transaction's inputs, ensuring consistency between validation and execution. The \textbf{stateful step} takes the validation token $\tau_T$ along with the object versions assigned to $T$ by the \router, then performs state reads and writes before terminating.

When $\mathrm{Validate}(T)$ returns invalid, the stateful step still executes but performs a no-op operation. This no-op consumes the assigned versions to preserve the per-key version streams, ensuring that version numbers remain consistent across all transactions.

\subsection{Time and Consensus Interface}
The system distinguishes between speculative and committed work through its consensus interface. Pre-consensus work on proposals is speculative and maintains the invariant of being side-effect free on state. The system uses consensus rounds to define logical epochs $E\in\mathbb{N}$, which serve as boundaries for lease expiration and provide a synchronization point for system-wide operations.

\subsection{Communication}
Communication between \workers is controlled to maintain system integrity. \workers may engage in peer-to-peer communication solely for fetching object state when authorized by the \router through a lease transfer. This restriction ensures that all state movement is tracked and authorized centrally while allowing direct transfers between execution nodes.

\subsection{Version Assignment (\router)}\label{sec:version-assignment}
The \router assigns versions to transactions deterministically based on their position in the consensus order. For a sequence of transactions $\langle T_1,\dots,T_n\rangle$, the \router performs version assignment by walking through the order exactly once. For each transaction $T_i$ and for each key $k\in R(T_i)\cup W(T_i)$, the \router increments the version counter $\mathrm{ver}[k]\leftarrow\mathrm{ver}[k]+1$ and assigns the input version $(k,\mathrm{ver}[k])$ to transaction $T_i$. This mechanism ensures the property of \emph{single-consumption}: each object version $(k,i)$ is consumed by exactly one transaction, after which it becomes eligible for garbage collection.

To enable speculative execution during the consensus window, the \router applies the same version assignment rule to proposed transaction orders using separate speculative counters $\mathrm{ver}^{\text{spec}}[\cdot]$. If the committed order matches the proposal, the speculative assignment is adopted, allowing transactions to proceed without recomputation. Otherwise, the speculative assignment is discarded and the \router recomputes versions based on the committed order.

\subsection{Ownership and Leasing}
The system maintains an ownership model where each key $k$ has at most one owner \worker at any given time. The \router serves as the authoritative source for ownership metadata, maintaining $\mathrm{owner}[k]$ for every key in the system to ensure consistency and prevent conflicting ownership claims.

Lease transfers enable state movement between workers when transaction placement requires it. When a transaction $T$ is placed on \worker $A$ but requires a key $k$ currently owned by \worker $B$, the \router orchestrates the transfer by including a source hint $B$ in the dispatch message to $A$ and atomically updating the ownership metadata $\mathrm{owner}[k]\leftarrow A$. \worker $A$ then fetches the required version $(k,i')$ where $i'\le i$ from \worker $B$ through peer-to-peer communication. If \worker $B$ has already checkpointed its state or is unavailable, \worker $A$ fetches the data from the \router.

At epoch boundaries, the system performs ownership reconciliation. Every worker takes a snapshot of all objects it owns that were modified during epoch $E$ and reports this information to the \router. This reporting ensures that the \router maintains a view of the global state distribution for garbage collection and recovery operations.

\subsection{Persistence and Recovery}
At each epoch $E$, the \router constructs a global snapshot by combining lease-aligned reports from all \workers. These reports contain the latest versions of objects owned by each worker that were modified during the epoch. After the \router persists the snapshot and acknowledges receipt, workers may garbage collect versions with indices less than or equal to $E$ according to the configured policy.

When a \worker crashes, the \router selects the latest snapshot epoch $E$, spawns replacement workers, and seeds them with state from the snapshot. The replacement workers then replay the committed transaction log from round $E\!+\!1$ to the current round, reconstructing the lost state through deterministic execution.

Garbage collection occurs at two points. During normal execution, when a worker creates a new version of an object, it may garbage collect all older versions of that object. At epoch boundaries, after sending its report to the \router and receiving acknowledgment, a worker may garbage collect all but the latest version of each object it owns.

\subsection{Execution Rule at \workers}\label{sec:execution-rule}
A \worker executes a stateful transaction $T$ only when all \router-assigned input versions $(k,i)$ for $k\in R(T)\cup W(T)$ are available locally, either produced by prior local execution or fetched through lease transfer. This availability requirement ensures that transactions execute with the input versions determined by the global ordering.

When reading version $(k,i)$, the worker returns the value produced by the unique transaction to access $k$ preceding $T$ in the committed order. If no such transaction exists, the read returns the value from the snapshot or the initial value. Executing $T$ consumes all assigned input versions $(k,i)$ and produces new versions $(k,i\!+\!1)$ for each $k\in W(T)$ with the updated values.

Before consensus commitment, \workers may execute only stateless validation. No state reads or writes, lease transfers, or version consumption occur during the pre-consensus phase, maintaining the invariant that speculative work has no side effects on persistent state.

\subsection{Notation Summary}
\begin{itemize}[leftmargin=12pt, itemsep=2pt]
  \item $\mathcal{K}$: The set of all keys in the system
  \item $\mathcal{V}$: The domain of values that objects can hold
  \item $R(T)$: The declared read set of transaction $T$
  \item $W(T)$: The declared write set of transaction $T$
  \item $\mathrm{ver}[k]$: The next version index for key $k$ maintained at the \router
  \item $(k,i)$: An object version, pairing key $k$ with version index $i$
  \item $\mathrm{owner}[k]$: The current owner \worker of key $k$
  \item $E$: An epoch corresponding to a consensus round
  \item $\tau_T$: The validation token for transaction $T$
\end{itemize}

\section{Assumptions}\label{sec:assumptions}

\system makes the following assumptions to ensure correctness:

\paragraph{Core System Assumptions}
\begin{itemize}[leftmargin=12pt, itemsep=2pt]
  \item \textbf{A1 (Consensus)}: The consensus layer provides a total order over transactions in committed blocks with safety (no forks) and eventual liveness.
  \item \textbf{A2 (Determinism)}\label{asm:determinism}: Stateless and stateful execution are deterministic given their inputs, with no hidden sources of nondeterminism (e.g., system time).
  \item \textbf{A3 (Cryptography)}: Tokens produced by $\mathrm{Validate}(T)$ are unforgeable, bound to $T$'s inputs, and verifiable by the stateful step.
\end{itemize}

\paragraph{Communication and Timing}
\begin{itemize}[leftmargin=12pt, itemsep=2pt]
  \item \textbf{A4 (Eventual Delivery)}\label{asm:eventual-delivery}: Messages are eventually delivered despite possible delay, reordering, or duplication; retry ensures eventual delivery.
  \item \textbf{A5 (Local Liveness)}\label{asm:local-liveness}: If a worker starts executing the stateless or stateful part of a transaction and does not crash, it eventually completes execution.
  \item \textbf{A6 (Partial Synchrony)}: The system is partially synchronous with an eventually perfect failure detector: crashed workers are eventually suspected (completeness) and correct workers are eventually not suspected (accuracy).
\end{itemize}

\paragraph{Failure Model}
\begin{itemize}[leftmargin=12pt, itemsep=2pt]
  \item \textbf{A7 (Worker Failures)}: Workers are crash-stop and may be replaced. No Byzantine behavior occurs within a validator.
  \item \textbf{A8 (Stability)}\label{asm:stability}: Eventually the inter-crash interval is sufficiently long to allow recovery to complete.
  \item \textbf{A9 (Router Durability)}: The \router remains available and its storage persists across worker failures. The \router is a single point of failure.
\end{itemize}

\paragraph{Performance Optimizations.}
The following assumption enables speculation but is not required for correctness:
\begin{itemize}[leftmargin=12pt, itemsep=2pt]
  \item \textbf{Proposal Visibility}: The \router may observe proposed blocks and their intra-block order before commitment.
\end{itemize}

\section{Pseudocode}

Algorithms~\ref{alg:router}--\ref{alg:elasticity} show the main functions of the \router and \workers. To simplify the presentation, we assume the existence of some black-box helper functions. These are typeset in \texttt{monospace}. To distinguish the functions we define from the black-box helpers, we typeset the former in \textsc{Small Caps}.

\begin{algorithm}[t]
\caption{\router main functions}
\label{alg:router}

\Proc{\OnProposedBlock{block}}{
  \RunStateless{block}\; \label{line:router-runstateless}
  specPlan $\gets$ \PlanBlock{block}\; \label{line:router-planblock-spec}
}
\Proc{\OnCommittedBlock{block}}{
  \If{\MatchesProposal{block}}{plan $\gets$ specPlan\; \label{line:router-adopt-spec}}
  \Else{plan $\gets$ \PlanBlock{block}\; \label{line:router-recompute}}
  \Dispatch{plan}\; \label{line:router-dispatch}
}
\Proc{\PlanBlock{block}}{
  vers $\gets$ \AssignVersions{block}\;  deps $\gets$ \BuildDependencies{block}\; \label{line:router-assign-build}
  units $\gets$ \Partition{deps}\;       place $\gets$ \PlaceUnits{units}\; \label{line:router-partition-place}
  \Return $\{\text{vers},\text{place}\}$\;
}
\Proc{\RunStateless{block}}{
  \ForEach{txn $\in$ block}{\ScheduleStateless{txn}\; \label{line:router-schedule-stateless}}
}
\Proc{\Dispatch{plan}}{
  \ForEach{txn $\in$ plan.block}{ \label{line:router-dispatch-foreach}
    vers $\gets$ \PlanVers{txn}\; place $\gets$ \PlanPlace{txn}\; sources $\gets$ \PlanSources{txn}\; \label{line:router-plan-per-tx}
    \ForEach{(k,i) $\in$ vers}{\SetOwner{k, place}\; \label{line:router-setowner}}
    \SendToWorker{txn, vers, place, sources}\; \label{line:router-send}
  }
}
\end{algorithm}

\begin{algorithm}[t]
\caption{\worker main functions}
\label{alg:worker}
\Proc{\OnStateless(txn)}{token $\gets$ \Validate{txn}\;\label{line:worker-publish-token}}
\Proc{\OnStateful(txn, versions, sources)}{
  \EnsureOwnership{versions, sources}\; \label{line:worker-ensure-ownership}
  \AwaitAllAvailable{versions, sources}\; \label{line:worker-await-versions} token $\gets$ \AwaitToken{txn}\; \label{line:worker-await-token}
  \If{\Invalid{token}}{\ConsumeNoOp{versions}\; \Return\;} \label{line:worker-consume-noop}
  readSet $\gets$ \Read{versions.read}\;  newVals $\gets$ \Execute{txn, readSet, token}\; \label{line:worker-read-exec}
  \Write{versions.write, newVals}\; 
}
\Proc{\EnsureOwnership(versions, sources)}{
  \ForEach{(k,i) $\in$ versions}{ \label{line:worker-ensure-foreach}
    \While{not \HasVersion{self, k, i}}{ \label{line:worker-ensure-while}
      hint $\gets$ sources[k]\;  verReq $\gets$ \LatestLeq{i}\; \label{line:worker-ensure-hint}
      ok $\gets$ \PeerFetch{self, hint, k, verReq}\; \label{line:worker-peerfetch}
      \If{not ok}{\FetchFromRouter{self, k, verReq}\; \label{line:worker-fetch-from-router}}
    }
  }
}
\end{algorithm}


\begin{algorithm}[t]
\caption{Persistence and Recovery (on \router)}
\label{alg:persistence}
\Proc{\OnEpochTrigger(E)}{\RequestSnapshot{E}\; \CollectAndPersist{E, reports}\; \label{line:persist-snapshot}}
\Proc{\OnWorkerFailure(failed)}{ \label{line:persist-onworkerfailure}
  E $\gets$ \LatestSnapshotEpoch{}\;  new $\gets$ \SpawnWorkers{}\;  \Seed{new, snapshot[$E$]}\; \label{line:persist-spawn-seed}
  range $\gets$ \LogRange{E+1, head}\;  \Replay{new, range}\;  \PromoteReplacement{new}\; \label{line:persist-replay-promote}
}
\end{algorithm}

\begin{algorithm}[t]
\caption{Elasticity (on \router)}
\label{alg:elasticity}
\Proc{\ScaleOut()}{w $\gets$ \SpawnWorker{}\; \BiasStateless{w}\; \GraduallyPlaceStateful{w}\;}
\Proc{\ScaleIn(worker)}{\StopPlacingStateful{worker}\; \BiasStateless{worker}\; \WaitForLeasesToTransferOrExpire{}\; \Shutdown{worker}\;}
\end{algorithm}

\paragraph{Helpers (black boxes)}
\begin{itemize}[leftmargin=12pt, itemsep=2pt]
  \item \texttt{AssignVersions}: Assign per-key version numbers to each transaction in block order.
  \item \texttt{BuildDependencies}: Construct the intra-block dependency graph from read/write version relationships.
  \item \texttt{Partition}: Group the dependency graph into execution units that preserve dependencies.
  \item \texttt{PlaceUnits}: Map execution units to workers according to placement policy (load/locality/balance).

  \item \texttt{ScheduleStateless}: Enqueue stateless validation tasks for execution (pre-commit, no state).
  \item \texttt{Validate}: Run the stateless function; return a token or invalid.
  \item \texttt{PublishToken}: Make a transaction's validation token available to consumers.
  \item \texttt{AwaitToken}: Block until the transaction's token is available (or invalid).
  \item \texttt{AwaitAllAvailable}: Ensure all required versions are locally available by driving peer fetches using source hints and falling back to the \router.

  \item \texttt{ConsumeNoOp}: Consume assigned versions without changing state when validation is invalid.
  \item \texttt{Read}: Materialize values for the assigned read-set versions from local storage.
  \item \texttt{Execute}: Run the stateful logic using the read set and token to compute new values.
  \item \texttt{Write}: Persist new values to the assigned write-set versions and make them available.
  \item \texttt{Ack}: Acknowledge transaction completion to the coordinator.

  \item \texttt{Owner}: Return the current worker that holds the lease for a key.
  \item \texttt{SetOwner}: Update the authoritative owner of a key after transfer.
  \item \texttt{HasVersion}: Check if a worker holds a specific key version locally.
  \item \texttt{AuthorizeTransfer}: Approve a lease transfer from current to target owner. (\router-side metadata only.)
  \item \texttt{PeerFetch}: \worker-to-\worker transfer of the requested key state ($\le$ requested version) using the \router-provided source hint.
  \item \texttt{FetchFrom\router}: Fallback transfer of the requested key state ($\le$ requested version) from the \router when the current owner has checkpointed or is unavailable.

  \item \texttt{RequestSnapshot}: Instruct workers to snapshot at epoch $E$ and report. Ack receipt so that workers can garbage collect.
  \item \texttt{CollectAndPersist}: Aggregate snapshot reports and persist a global snapshot.
  \item \texttt{LatestSnapshotEpoch}: Return the latest fully persisted snapshot epoch.

  \item \texttt{SpawnWorkers}: Start replacement worker processes for recovery.
  \item \texttt{Seed}: Load snapshot state into newly spawned workers.
  \item \texttt{Replay}: Deterministically re-execute committed log entries since the snapshot.
  \item \texttt{PromoteReplacement}: Promote a single recovered worker to replace the failed one and shut down the rest.
  \item \texttt{Shutdown}: Cleanly stop a worker after it is drained.

  \item \texttt{BiasStateless}: Prefer routing stateless tasks to the specified worker.
  \item \texttt{GraduallyPlaceStateful}: Incrementally assign stateful work/leases to the worker.
  \item \texttt{StopPlacingStateful}: Stop assigning new stateful work to the worker.
  \item \texttt{WaitForLeasesToTransferOrExpire}: Wait until the worker's leases move or expire at epochs.
\end{itemize}

\section{Correctness Properties}

\system satisfies the following correctness properties:

\begin{theorem}[Consensus-ordered serializability]\label{thm:serializability}
For a correct validator, execution under \system of the transactions $\{ T_1,\dots,T_n \}$ yields the same state as if the transactions were executed sequentially, in consensus order.
\end{theorem}

\begin{theorem}[Cross-validator determinism]\label{thm:determinism}
No two correct validators that execute the same sequence of transactions $\langle T_1,\dots,T_n\rangle$ yield different states.
\end{theorem}

\begin{theorem}[Liveness]\label{thm:liveness}
For every committed $T_i$, \system eventually executes $T_i$ and incorporates its effect into the global state.
\end{theorem}





\section{Proofs}

\subsection{Serializability}

We prove Theorem~\ref{thm:serializability}. Let $\langle T_1,\dots,T_n\rangle$ be the total transaction order produced by consensus. Let $R(T)$ and $W(T)$ denote the read/write sets of transaction $T$ (\S\ref{sec:system-model}). The \router assigns versions according to \S\ref{sec:version-assignment} and \workers execute according to \S\ref{sec:execution-rule}. We begin with the following helpful invariants.

\begin{lemma}[Per-key versions are consecutive]\label{lem:per-key-linearity}
For every key $k\in\mathcal{K}$, the sequence of versions $(k,0),(k,1),\dots$ is a linear chain with no gaps (i.e., the versions are consecutive), and each version is assigned at most once in the consensus order of the transactions that touch $k$. 
\end{lemma}
\begin{proof}
This follows directly from the \router's assignment rule: for each $T_i$ and each $k\in R(T_i)\cup W(T_i)$, the \router increments $\mathrm{ver}[k]$ once and assigns the resulting $(k,\mathrm{ver}[k])$ to $T_i$ (\S\ref{sec:version-assignment}), in consensus order.
\end{proof}

\begin{lemma}[Read from latest access]\label{lem:reads-latest}
  When a transaction $T$ reads key $k$, it observes the value produced by the latest transaction that precedes $T$ in the consensus order and accesses $k$, or the initial value if no such transaction exists.
\end{lemma}
\begin{proof}
  Let $T'$ be the latest transaction that precedes $T$ in the consensus order and accesses $k$, or the initial state if no such transaction exists. By the version assignment rule (\S\ref{sec:version-assignment}), versions are assigned consecutively in consensus order. Thus, $T$ reads the value produced by $T'$. 
\end{proof}

\begin{lemma}[No pre-commit state effects]\label{lem:no-precommit}
Pre-consensus actions do not read or mutate state, nor consume versions or transfer leases, hence cannot affect post-commit execution results.
\end{lemma}
\begin{proof}
By \S\ref{sec:execution-rule}, \workers only perform stateless validation before consensus (\algline{alg:router}{line:router-schedule-stateless}); they do not perform any state reads/writes, transfer any leases, or consume  any versions before commitment.
\end{proof}

\begin{lemma}[Speculative assignment is safe]\label{lem:spec-safe}
The version assignment of the \router equals that induced by the committed order.
\end{lemma}
\begin{proof}
If a proposed block's order matches the committed order, adopting the speculative version plan is equivalent to recomputing on the committed order; otherwise, speculation is discarded and the committed-order plan is used. In either case, the effective version assignment equals that induced by the committed order.
\end{proof}

\begin{lemma}[Version value uniqueness]\label{lem:version-uniqueness}
For any key $k$ and version index $i$, if version $(k,i)$ exists at multiple locations in the system (e.g., at different workers, in persistent storage, or produced by re-execution during recovery), all copies have identical values.
\end{lemma}
\begin{proof}
By contradiction. Suppose there exist two copies of version $(k,i)$ with different values $v$ and $v'$. Let $(k,i)$ be the earliest such version (in consensus order of producing transactions) where a disagreement occurs.

By the single-consumption property (Lemma~\ref{lem:per-key-linearity}), exactly one transaction $T$ in the consensus order produces version $(k,i)$ by consuming $(k,i-1)$ and writing to $k$.

Transaction $T$ may be executed multiple times (e.g., once before a failure and once during recovery). In each execution, $T$ reads the same input versions: for each $k' \in R(T)$, it reads version $(k',j_{k'})$ where $j_{k'}$ is determined by the version assignment. By minimality of $(k,i)$, all versions earlier than $(k,i)$ in consensus order have unique values, so $T$ observes identical input values in each execution.

Since $T$ has the same validation token $\tau_T$ (determined by consensus) and reads identical input values, by the determinism assumption (\S\ref{sec:assumptions}), $T$ must produce the same output value for key $k$. Therefore, both copies of $(k,i)$ must have the same value. Contradiction.
\end{proof}

\paragraph{Constructing the sequential history.}
Let $\mathsf{H}$ be the execution history in which transactions are applied sequentially in consensus order. We show by induction on $i$ that the state after completing $T_i$ in \system equals the state after applying $T_1,\dots,T_i$ sequentially in $\mathsf{H}$, assuming the same validation outcomes for $T_1,\dots,T_i$ in both cases. Note that by Lemma~\ref{lem:version-uniqueness}, even if failures and recovery cause transactions to be re-executed, each version has a unique, well-defined value throughout the system, ensuring that the state remains consistent despite potential re-executions.

\begin{lemma}[Prefix equivalence]\label{lem:prefix}
For all $i\in\{0,\dots,n\}$, after \system has executed $T_1,\dots,T_i$, the value of every key version $(k,j)$ equals that in the sequential execution of $T_1,\dots,T_i$.
\end{lemma}
\begin{proof}
By induction on $i$. Base $i=0$ holds because the initial state is the same. For the inductive step, fix $i>0$ and assume the claim for $i-1$. Consider $T_i$:
\begin{itemize}
  \item If $\mathrm{Validate}(T_i)$ is invalid, \system executes a no-op that consumes assigned versions (\S\ref{sec:transactions-separation}) so the state remains unchanged; the sequential execution also skips writes. By Lemma~\ref{lem:spec-safe}, the version assignment produced by the \router equals the one induced by consensus. Thus, the state produced by \system coincides with the state produced by the sequential execution after $T_i$.
  \item If $\mathrm{Validate}(T_i)$ is valid, by Lemma~\ref{lem:reads-latest}, the read set observed by $T_i$ in \system equals the values accessed by the latest predecessors in $\langle T_1,\dots,T_{i-1}\rangle$ (or initial values), which by the inductive hypothesis equal those in the sequential prefix state. By the determinism assumption (\S\ref{sec:assumptions}) and Lemma~\ref{lem:spec-safe}, the stateful computation produces the same outputs as the sequential execution would on the same inputs. \system then writes them to $(k,j+1)$ for $k\in W(T_i)$ and consumes $(k,j)$ as specified, matching the sequential step. Keys not touched by $T_i$ are unchanged. 
\end{itemize}
Thus, the state produced by \system coincides with the state produced by the sequential execution after $T_i$, completing the inductive step.
\end{proof}

\begin{proof}[Proof of Theorem~\ref{thm:serializability}]
By Lemma~\ref{lem:prefix} with $i=n$, the state after completing all transactions in \system equals the state after sequential execution in consensus order. Therefore, \system is serializable with respect to the consensus order.
\end{proof}

\subsection{Determinism}

We derive Theorem~\ref{thm:determinism} as a corollary of serializability and the determinism assumption (\S\ref{sec:assumptions}).

\begin{proof}[Proof of Theorem~\ref{thm:determinism}]
Let $V$ and $V'$ be two correct validators that execute the same committed sequence $\langle T_1,\dots,T_n\rangle$ from the same initial snapshot. By Theorem~\ref{thm:serializability}, each validator's execution is equivalent to executing $\langle T_1,\dots,T_n\rangle$ sequentially in consensus order. By the determinism assumption, stateless validation outcomes and stateful effects are deterministic functions of their inputs and reads; therefore both sequential executions produce identical per-transaction effects (invalid transactions perform consuming no-ops) and thus identical final states. Hence $V$ and $V'$ cannot diverge.
\end{proof}

\subsection{Liveness}

We prove Theorem~\ref{thm:liveness}. Intuitively, we first show that \system{'s} scheduler is fair and then, given fairness at the \router, eventual delivery in the network, and partial synchrony, we show that every committed transaction is eventually dispatched, validated, and its stateful step executed, despite worker failures.

\subsubsection{Fairness}
\begin{theorem}[Scheduler fairness]\label{thm:scheduler-fairness}
  Fix any finite committed block $B$. After \textsc{OnCommittedBlock}$(B)$ the \router eventually dispatches every transaction $T \in B$ to some \worker.
\end{theorem}
  
\begin{proof}
  Upon commitment of $B$, the \router either adopts the precomputed stateful plan (if the proposal order matches) or recomputes it on the committed order. In either case, planning proceeds as follows: it builds the dependency graph $D(B)$ and extracts its disconnected components $\mathcal{G}$, which form a partition of the finite set $B$. It then iterates once over each $G \in \mathcal{G}$, computes a score $S_i$ for every worker $i$ from a finite worker set, and selects $i^\star=\arg\max_i S_i$ (ties broken arbitrarily), thereby assigning $G$—and thus every $T \in G$—to $i^\star$. Because both $\mathcal{G}$ and the worker set are finite, this placement completes in finite time and yields a target worker for every $T \in B$.
  
  The \router then executes dispatch over the finite set $B$, invoking \texttt{SendToWorker} once per $T$. By eventual delivery, each dispatch message reaches its assigned \worker. Therefore, every committed transaction $T \in B$ is eventually dispatched, establishing scheduler fairness.
\end{proof}

\subsubsection{Liveness Given Fairness}

\begin{remark}[No circular dependencies]\label{obs:no-circular}
  Transactions cannot deadlock due to circular version dependencies. 
\end{remark}
\begin{proof}
The version assignment rule ensures that transaction $T_i$ at position $i$ in consensus order $\langle T_1,\dots,T_n\rangle$ can only depend on versions produced by transactions $T_j$ where $j < i$. This creates a directed acyclic graph where dependency edges point strictly backward in consensus order. Since the dependency relation respects the total order of consensus, cycles are impossible.
\end{proof}

\begin{proof}[Proof of Theorem~\ref{thm:liveness}]
  We consider the failure-free case first and later extend to the general case. Let $\langle T_1,T_2,\dots\rangle$ be the transaction order output by consensus. We prove by induction on $i\geq 1$ that $T_i$ eventually executes.
  
  Base $i=1$. By scheduler fairness and eventual delivery (\S\ref{sec:assumptions}), the \router eventually dispatches the stateless and stateful steps of $T_1$ to some \worker or \workers. 
  By local liveness, the stateless step of $T_1$ eventually completes and its validation token becomes available to the stateful step. 
  For each required input key $k\in R(T_1)\cup W(T_1)$, the corresponding version is the initial version $(k,0)$, which is available at the \router's persisted state. The \worker requests these versions if necessary and eventually acquires them (\algline{alg:worker}{line:worker-ensure-ownership}--Line~\ref{line:worker-await-versions}), as well as the token (\algline{alg:worker}{line:worker-await-token}).
  With the token and input versions available and by local liveness, the \worker eventually executes the stateful step (or the consuming no-op, if invalid), completing $T_1$.
  
  Inductive step. Assume $T_1,\dots,T_{i-1}$ eventually execute, and consider $T_i$. By fairness and eventual delivery, $T_i$ is dispatched to some \worker. 
  The stateless step of $T_i$ is scheduled and, by local liveness, its validation token is eventually computed (\algline{alg:worker}{line:worker-publish-token}) and obtained by the \worker assigned the stateful step (\algline{alg:worker}{line:worker-await-token}). 
  For each key $k$ that $T_i$ requires as input, let $T'$ be the latest transaction to access $k$ preceding $T_i$ in consensus order (if none such transaction exists, interpret $T'$ as the creation of the  initial state).
  By the induction hypothesis, $T'$ has completed (note that $T'$ precedes $T_i$, so no circular dependency exists by Remark~\ref{obs:no-circular}), thus the required version for $k$ exists at its owner (or at the \router's persisted state). 
  The \worker requests this version if necessary and eventually acquires it (\algline{alg:worker}{line:worker-ensure-ownership}--Line~\ref{line:worker-await-versions}).
  With the token and inputs available, by local liveness, the \worker eventually performs the stateful step (or the consuming no-op if invalid), completing $T_i$. This completes the induction.

  We now allow \worker failures (crash-stop) and argue that failures do not prevent the eventual completion of any committed $T_i$. Assume that a worker $W$ executing the stateless or stateful step of $T_i$ crashes. The \router detects the failure and initiates recovery (\algline{alg:persistence}{line:persist-onworkerfailure}). It spawns replacements and seeds them from the latest persisted snapshot (\algline{alg:persistence}{line:persist-spawn-seed}), which cannot include $T_i$, since $T_i$ has not completed yet. The \router then deterministically replays the committed log up to the current head (\algline{alg:persistence}{line:persist-replay-promote}), after which it promotes the replacement worker $W'$ and shuts down the rest. It is possible that \worker crashes interrupt the recovery process, but by the stability assumption (\S\ref{sec:assumptions}), there will eventually be a long enough failure-free interval to complete the recovery and finish executing $T_i$.
\end{proof}

\end{document}